\documentclass[journal]{IEEEtran}

\usepackage{cite}
\usepackage{amsthm}
\usepackage{amsmath,amssymb,amsfonts}
\usepackage{algorithmic}
\usepackage{graphicx}
\usepackage{textcomp}
\usepackage{url}
\usepackage{amsthm}
\usepackage{multirow}
\usepackage{multicol}

\usepackage[table]{xcolor}

\usepackage{algorithm}
\usepackage{algorithmic}

\def\BibTeX{{\rm B\kern-.05em{\sc i\kern-.025em b}\kern-.08em
    T\kern-.1667em\lower.7ex\hbox{E}\kern-.125emX}}

\newtheorem{proposition}{Proposition}

\newtheorem{definition}{Definition}

\begin{document}

\title{Joint Auction-Coalition Formation Framework for Communication-Efficient Federated Learning in UAV-Enabled Internet of Vehicles}

\author{
Jer Shyuan Ng\thanks{JS.~Ng and WYB.~Lim are with Alibaba Group and Alibaba-NTU Joint Research Institute, Nanyang Technological University, Singapore. }, 
Wei Yang Bryan Lim, 
Hong-Ning Dai\thanks{H.~Dai is with Faculty of Information Technology, Macau University of Science and Technology, Macau SAR.},
Zehui Xiong\thanks{Z.~Xiong is with Alibaba-NTU Joint Research Institute, and also with School of Computer Science and Engineering, Nanyang Technological University, Singapore. }, 
Jianqiang Huang,\thanks{JQ.~Huang and XS.~Hua are with Alibaba Group. }
Dusit~Niyato,\thanks{D.~Niyato is with School of Computer Science and Engineering, Nanyang Technological University, Singapore. }~\textit{IEEE~Fellow}, 
Xian-Sheng~Hua,~\textit{IEEE Fellow}, 
Cyril Leung,\thanks{C. Leung is with The University of British Columbia and Joint NTU-UBC Research Centre of Excellence in Active Living for the Elderly (LILY).}
Chunyan Miao \thanks{C.~Miao is with Joint NTU-UBC Research Centre of Excellence in Active Living for the Elderly (LILY) and School of Computer Science and Engineering, Nanyang Technological University, Singapore.}
}


\maketitle

\begin{abstract}

Due to the advanced capabilities of the Internet of Vehicles (IoV) components such as vehicles, Roadside Units (RSUs) and smart devices as well as the increasing amount of data generated, Federated Learning (FL) becomes a promising tool given that it enables privacy-preserving machine learning that can be implemented in the IoV. However, the performance of the FL suffers from the failure of communication links and missing nodes, especially when continuous exchanges of model parameters are required. Therefore, we propose the use of Unmanned Aerial Vehicles (UAVs) as wireless relays to facilitate the communications between the IoV components and the FL server and thus improving the accuracy of the FL. However, a single UAV may not have sufficient resources to provide services for all iterations of the FL process. In this paper, we present a joint auction-coalition formation framework to solve the allocation of UAV coalitions to groups of IoV components. Specifically, the coalition formation game is formulated to maximize the sum of individual profits of the UAVs. The joint auction-coalition formation algorithm is proposed to achieve a stable partition of UAV coalitions in which an auction scheme is applied to solve the allocation of UAV coalitions. The auction scheme is designed to take into account the preferences of IoV components over heterogeneous UAVs. The simulation results show that the grand coalition, where all UAVs join a single coalition, is not always stable due to the profit-maximizing behavior of the UAVs. In addition, we show that as the cooperation cost of the UAVs increases, the UAVs prefer to support the IoV components independently and not to form any coalition.

\end{abstract}


\section{Introduction}
It is expected by 2023, the Internet of Things (IoT) will have a market size of \$1.1 trillion \cite{IDC}. With an increasing number of vehicles being connected to the IoT, instead of just providing information to the drivers and uploading the sensor data to the Internet, from which other users can access the data, the vehicles are connected to each other in a network such that they are able to exchange their sensor data among each other for the optimization of a well-defined utility function.
As a result, the deep integration between IoT and vehicular networks has led to the evolution of the Internet of Vehicles (IoV) \cite{zhuang2020sdn} from the traditional Vehicular Ad-Hoc Networks (VANETs).

IoV is an open and integrated network with communication, computation, intelligence, storage and learning capabilities. Different from VANETs of which the performance is constrained by the number of vehicles connected to the networks and the mobility of the vehicles, the emphasis of the IoV network on information interaction among users, vehicles and Roadside Units (RSUs) has enabled many transportation-related applications that improve the efficiency of the transportation systems, upgrade the service level of the cities and provide comfort and higher satisfaction to the drivers. The mobile crowd sensing paradigm can make use of the sensor-equipped vehicles, the RSUs as well as the mobile devices in the IoV network to produce useful traffic data that is transmitted to the cloud for analysis. The collection and the analysis of the crowd-sourced data allow us to gain meaningful insights for the provision of efficient traffic management \cite{zheng2019deep}, route planning \cite{route} as well as to promote safe driving \cite{safe}.

Given the enormous amounts of dynamic data generated in the IoV network, Artificial Intelligence technologies are leveraged to build data-driven models. Specifically, many machine learning algorithms are developed to process and analyze data in order to make smart traffic predictions and provide intelligent route recommendations. To build the data-driven models, the traditional machine learning approaches require the data to be transmitted to the central data server for training. As the number of components connected to the IoV network increases, the amount of data generated in the IoV network increases exponentially. The central data server does not have sufficient storage and computing resources to store and process such massive amount of data, limiting the performance of the trained models. Besides that, the transmission of raw data over the communication channels incurs high communication cost and risks the privacy of the data as it may be eavesdropped while being transmitted.

In face of these challenges, Federated Learning (FL) \cite{federated} has gained great interest of the research and industrial communities. Instead of training the machine learning model at the central data server, FL allows the training of the models to take place directly on the devices by leveraging on the growing computational and storage capabilities of the IoT devices. 

The IoV service provider runs an FL server, i.e., model owner, which employs multiple groups of parking vehicles or RSUs, i.e., workers, which collect information from their immediate surroundings. In each training iteration, the workers receive the model parameters from the model owner for model training on their respective local datasets. The updated local model parameters from all involved workers are then uploaded to the FL server which performs model aggregation. The model owner then transmits the updated model parameters to the workers for another iteration of local model training. A number of iterations is performed until a desired accuracy is achieved. Efficient data-driven models are built while preserving the privacy of the participating devices since only the model parameters, rather than the raw data, are exchanged.

However, there are several challenges pertaining to FL specifically that need to be addressed for large scale implementation of efficient FL. In particular, communication inefficiency is a key bottleneck in the FL process \cite{konen2016federated}. Since the federated network is connected to millions of devices, it is estimated that the time taken to update the local parameters over the communication channels takes longer than to compute the local parameters themselves. The data communication between the workers and the model owner may fail due to limited communication capacity and link failure \cite{Hu2013}, \cite{Abdullah2014}, which leads to the resulting node failure and missing node information. More than 30\% of node failure contributes to a significant increase of the average end-to-end delay and concurrent communication time \cite{li2018}. This means that the number of workers participating in the model training is reduced due to communication failure, thus causing a lower accuracy of the FL model. 

Therefore, we propose the use of Unmanned Aerial Vehicles (UAVs) to facilitate the network in achieving higher communication efficiency in the FL network. 

The vehicles and the RSUs may not have sufficient communication bandwidth and energy capacity to transmit the local model parameters directly to the FL server which is located far away. As a result, the FL server does not receive the local model parameters from these failing nodes, which cause a degraded performance in the trained FL model. The UAVs are increasingly used as relays between the ground base terminals and the network base station \cite{wirelessrelays} to improve network connectivity and extend coverage areas \cite{connect1}. Since UAVs can be quickly deployed to designated areas whenever they are needed, the UAVs are nearer to the vehicles and the RSUs. Instead of transmitting the local model parameters to the FL server that is further away, the vehicles and the RSUs only need to transmit them to the UAVs that are located nearer to them, of which the parameters are relayed to the FL server. Hence, as compared to the direct communication between the IoV components and the FL server, the UAVs help to improve this communication by reducing node failures which may be due to limited resources and energy capacity of the IoV components.

In this paper, we consider that the model owner is interested in facilitating FL among IoV components (Fig. \ref{system}), e.g., for traffic prediction. For efficient FL training, the model owner announces the FL task and the number of training iterations. The interested workers upload their sampling rate to the model owner. The model owner then selects its workers accordingly based on the importance of each worker, which is computed based on the uploaded sampling rates. Since the workers may not have sufficient resources, i.e., transmission power and communication bandwidth for efficient FL, the workers can utilize UAVs which have the required resources to facilitate the FL tasks. Firstly, the workers will conduct model training with their local datasets. Then, instead of sending the local model parameters to the remote FL server directly, the workers upload their local model parameters to the UAVs, which in turn aggregate the accumulated model parameters for relay to the model owner. Thereafter, the model owner aggregates the model parameters to derive the global FL model. The model owner will then announce the global FL model parameters to the workers for another round of FL training. The process continues until the number of iterations announced by the model owner is reached.

There are two advantages to this proposed approach. Firstly, the privacy of the workers is preserved since only local parameters of the workers are exchanged with the UAVs and the data of the workers remains locally at the workers' site. Secondly, with the help of UAVs, communication efficiency is increased, resulting in reduction of link and node failure. 

\begin{figure}
\includegraphics[width=\linewidth]{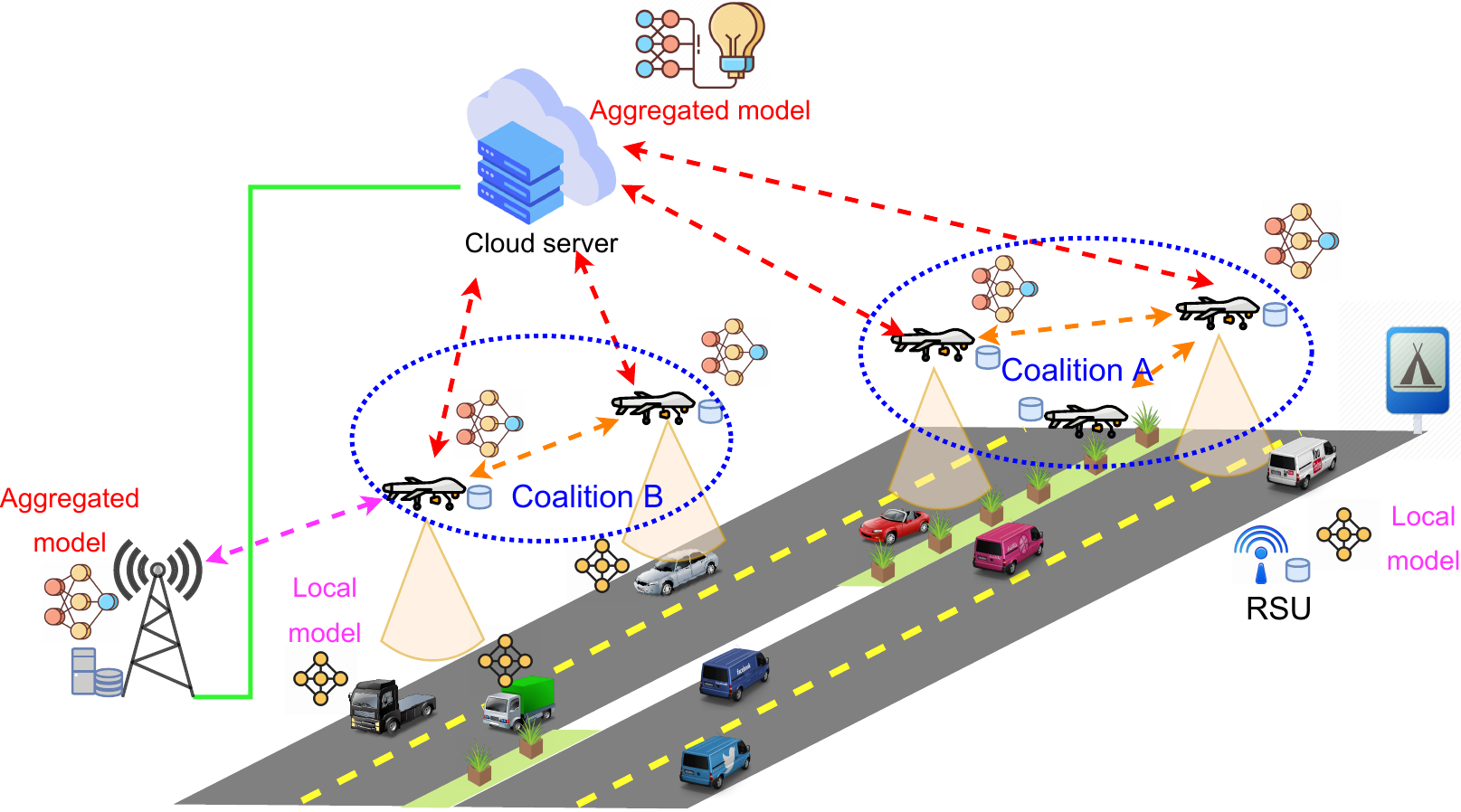}
\caption{System model consists of the cloud server (FL model owner), the vehicles and RSUs (selected FL workers) and the UAVs.}
\label{system}
\end{figure}

However, a single UAV may not have sufficient energy to last for the entire training duration and other resources for the FL training process. As such, the UAVs may cooperate with other UAVs to ensure that the edge aggregation and communications with the model owner are provided throughout the entire FL training process. To model the cooperation among the UAVs, we propose a coalition formation approach among cooperating UAVs. In order to incentivize UAVs to participate in the FL training, the UAVs are rewarded for their contributions if they complete the FL training process. The UAVs receive different payoffs for supporting FL training process for different cells of workers as the importance of the workers varies. In addition, the workers have preference for different UAVs. In order to elicit the valuations of cells of workers for each coalition of profit-maximizing UAVs, we propose an auction scheme. Each cell of workers submit their bids to each coalition of UAVs. Collectively, the profit-maximizing UAVs evaluate all possible coalitional structure formation and decide on the optimal allocation of UAVs. Since the UAVs are profit-maximizing, it is possible that there are two or more UAV coalitions which want to be allocated to the same cell of workers. In such a situation, based on the bids of the cells of workers, which also represent preferences of the cells of workers, the more preferred UAV coalition gets the allocation. 

The main goal of this work is to develop a joint auction-coalition formation framework of UAVs towards enabling efficient FL for IoV networks. The integration of coalition formations of UAVs and optimal bidding of cells of workers makes the proposed framework suitable for the FL training process as it ensures that the UAVs have sufficient resources to complete the FL tasks and at the same time, optimally allocates the profit-maximizing UAV coalitions to the groups of IoV components.  Our key contributions include:
\begin{enumerate}
\item We propose the introduction of UAVs as wireless relays in FL network where communication efficiency between the model owner and the workers is improved. 
\item We propose a joint auction-coalition formation framework where the problem of UAV coalitions allocation to the groups of IoVs components is formulated as a coalition formation game.
\item We propose a joint auction-coalition formation algorithm where a stable coalitional structure is achieved in which the allocation of profit-maximizing UAV coalitions to the groups of IoV components is solved by an auction scheme.

\end{enumerate}

The remainder of the paper is organized as follows: In Section~\ref{sec:related}, we review the related works. In Section~\ref{sec:sys}, we present the system model and the problem formulation. Our proposed approach of coalition formation of UAVs and auction design are discussed in Section~\ref{sec:coalitions} and Section~\ref{sec:auction}, respectively. Section~\ref{sec:sim} discusses the simulation results and analysis. Section~\ref{sec:conc} concludes the paper.

\section{Related Work}
\label{sec:related}

A number of recent studies investigate on the effective deployment of FL algorithms over the wireless networks. Several fundamental issues such as the optimal allocation of resources \cite{chen2019joint}, \cite{yang2019energy}, the development of efficient algorithms \cite{adaptive2019} and the design of incentive mechanisms \cite{hierarchical}, \cite{lim2020incentive} are current active research areas in FL. One of the main challenges in FL is communication inefficiency. In particular, the synchronous FL protocol implies that each training iteration only takes place as quickly as the slowest device, i.e., the straggler's effect \cite{lim2019federated}. As such, the dropped or straggling devices can lower the efficiency of FL training. The reduction of communication cost in the federated network setting can be achieved by reducing the total number of communication iterations or by reducing the size of transmitted data in each iteration \cite{li2019federated}. In order to minimize convergence time while optimizing performance, the study of \cite{chen2020convergence} looks into probabilistic user selection scheme and artificial neural networks. Given the faster convergence time, the end devices have greater capabilities to complete the FL training process. The study of \cite{konen2016federated} proposes to use structured and sketched updates to reduce communication cost in dealing with a large number of users over unreliable network connections. 


Apart from the studies that focus on the improvement of transmission efficiency in the wireless networks \cite{zhang2019artificial}, \cite{zhang2019deep}, due to the inherent attributes of UAVs such as flexibility, mobility and adaptive altitude,  there are an increasing number of UAV applications in wireless networks \cite{tutorial}. The growing capabilities of the UAVs have demonstrated their potential to be deployed in numerous applications, e.g., environmental sensing \cite{airsensing}, mobile edge computing \cite{mobileedge} and traffic surveillance \cite{puri2005survey}. In particular, UAVs are also increasingly being considered in IoV-related applications for more efficient Intelligent Transportation Systems (ITS) in the development of smart cities. The study of \cite{carcount} provides a solution in counting the number of cars based on images provided by the UAVs. The aerial images taken by the UAVs are used to track vehicles \cite{perez2014ground}. Moreover, UAVs are also used as aerial mobile base station for cellular traffic offloading \cite{cellular}. Several studies such as \cite{optimal2018} have been conducted on the optimal placement of the UAVs. 

However, to the best of our knowledge, few studies consider using UAVs for efficient FL. The involvement of UAVs as wireless relays between the IoV components and the FL server can greatly improve the communication efficiency of the FL training process by improving the connectivity and increasing the coverage of the IoV components. 

Due to the limitation of a single UAV in carrying out complex tasks, coalition formation games for tasks allocation have been investigated. The cooperation of multiple UAVs to maximize the coverage utility of the UAV network is considered in \cite{deployment2018}, which is useful in crowdsensing tasks to gather more data covering larger areas for more accurate analysis. A leader-follower approach \cite{leader2018} is also widely used by the UAVs to form multiple coalitions in order to complete designated tasks with minimal resources. The study in \cite{priority} considers coalition formations for tasks allocation taking into account tasks priority and requiring the UAVs to arrive simultaneously for task completion. However, these coalition formations approaches do not jointly consider the importance of the tasks that the UAVs perform and the optimal allocation of the UAVs to complete the tasks, subject to the resource constraints of the UAVs. Besides that, the preferences of the workers for heterogeneous UAVs are not considered.

To that end, in order to elicit the preferences of workers for the heterogeneous UAVs, we adopt an auction-based approach in our coalition formation design for task allocation. The use of sequential auction mechanism is proposed \cite{sequential2007} by considering the UAVs with limited communication range. However, it only considers one UAV in serving a task, whereas a single UAV may not be able to support a cell of workers for the entire FL training process in our model. The study of \cite{auction2016} explores the leader-follower approach where the leader UAVs auction for other UAVs to form coalitions that are adaptive to the changes in task workload and capabilities of the UAVs. The work of \cite{auctionprice} proposes a multi-layered cost computation method to calculate the price of the bid. While the study of auction design is well-explored in tasks allocation, where the bid winner is often determined by the bid with lowest cost or the bid with the highest utility, the valuation of a coalitional structure using an auction approach is often not considered.

As such, this inspires us to propose a joint auction-coalition formation framework (Fig. \ref{framework}) to determine the allocation of UAV coalitions to groups of IoV components for efficient FL over IoV networks.

\section{System Model}
\label{sec:sys}

\begin{figure*}
\includegraphics[width=\textwidth]{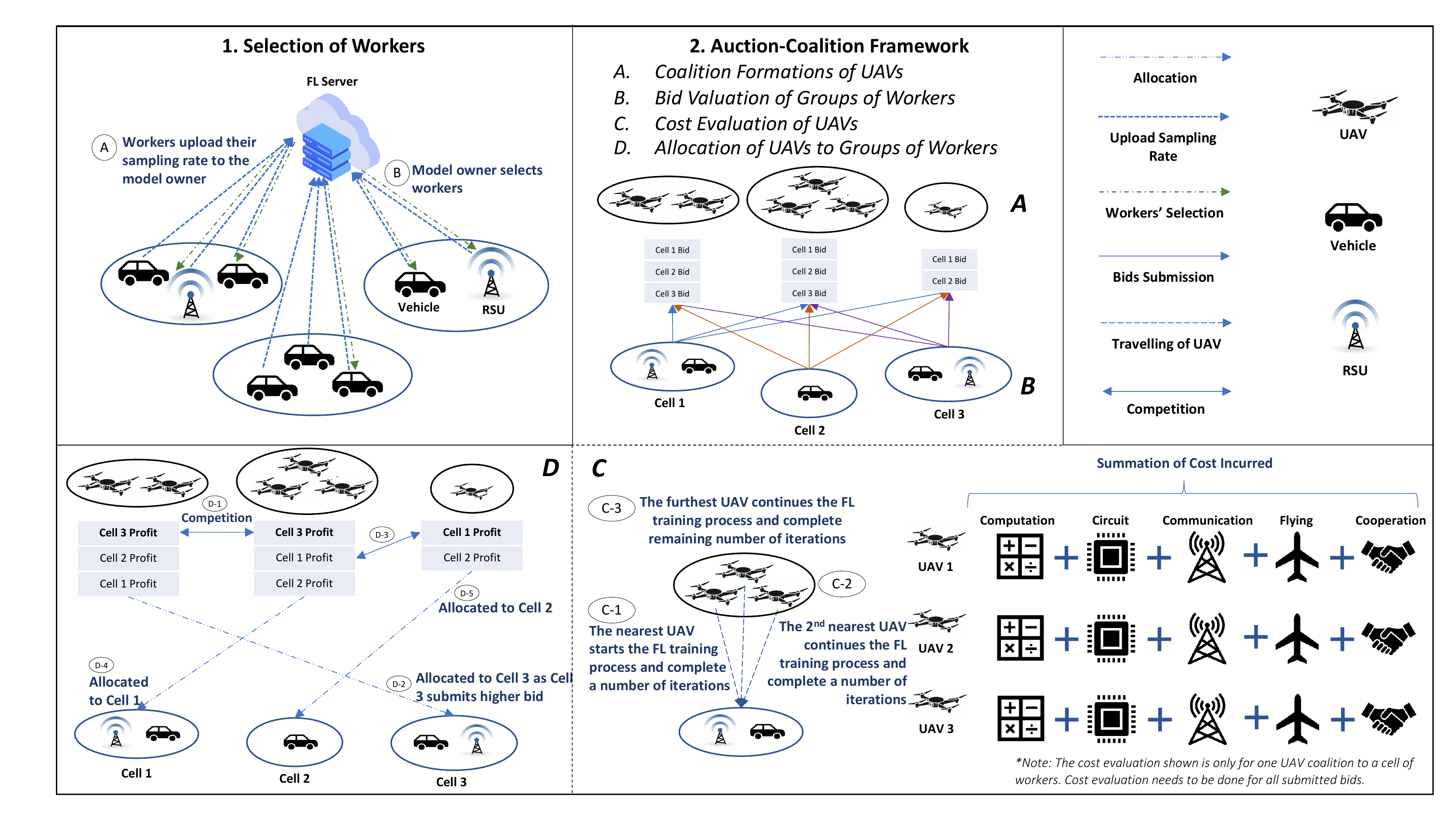}
\caption{Illustration of the Joint Auction-Coalition Formation Framework.}
\label{framework}
\end{figure*}

We consider a UAV-assisted FL network over the IoV paradigm. The system model is comprised of an FL model owner which announces crowdsensing tasks to a group of selected FL workers. The FL workers which complete the crowdsensing tasks are responsible for labelling the data and training the model on their own local datasets. In order to perform data labelling for model training, the FL workers can leverage the knowledge, i.e., labels from a source-domain party \cite{liu2018secure}. Thereafter, the local model parameters will be uploaded to the FL model owner for model aggregation. Since the communication efficiency between the model owner and the cell of workers is low due to link failure and missing nodes \cite{Hu2013}, \cite{Abdullah2014}, UAVs are introduced to facilitate the network to improve the communication efficiency of the FL model. 

We consider a coverage area that is divided into $I$ cells, represented by a set $\mathcal{I}=\{1,\ldots, i,\ldots, I\}$. Each cell $i$ contains $W_i$ workers, where the total number of workers in each cell $i$ is different and it is represented by $\mathcal{W}_i=\{1,\ldots, w_i,\ldots, W_i\}$. The term $d_{iw} $ denotes the $w^{\mathrm{th}}$ IoV component in cell $i$. Each worker has different levels of importance, denoted by $\sigma_{iw}$, e.g., based on the sampling rate and the quality of data collected. The workers with low sampling rate and low quality are of low importance as there is insufficient data for the FL training process. The workers upload their sampling rates to the model owner, which in turn computes the importance of each worker. Since the design for incentive mechanisms that guarantee the truthfulness of the workers' sampling rates is not the focus of this paper, existing incentive mechanisms such as the Vickrey-Clarke-Groves mechanism \cite{kantar2010incentive} can be implemented to ensure that the workers upload their true sampling rates. Based on the importance level of each worker, the model owner selects the workers to participate in the FL training process. After selecting the workers, the model owner announces an FL model that will be trained over $\mu$ iterations. The optimal value of $\mu$ is determined such that the global loss function is minimized given the resource-constrained workers \cite{adaptive2019}. Upon completion of the FL task, the model owner will make a payment to the participating workers.

In the network, there are $M$ geographically distributed UAVs, represented by a set $\mathcal{M}=\{1,\ldots, m,\ldots, M\}$ to facilitate the FL training process. The UAVs only collect the local model parameters from the individual workers after their model training and then perform edge aggregation before transmitting the aggregated local model parameters to the FL model owner. This ensures that the data belonging to the workers is not exchanged with any third parties, consequently preserving the data privacy. To improve communication efficiency, the UAVs need to ensure that they have sufficient energy to stay in the air throughout the entire FL process. In order to incentivise the UAVs to complete the FL training process, the workers will make a payment to the UAVs. In fact, the UAVs always maximize their payoffs. However, the individual UAVs may lack the energy capacities to stay in the air throughout the entire FL training process, and they cannot support certain cells of workers. As a result, the UAVs can consider to form coalitions. Note that each coalition of UAVs can only choose to facilitate the FL training process in one of the cells of workers. In order to elicit the willingness of payment of each cell of workers, an auction scheme is also explored to investigate the optimal allocation that maximizes the profits of the UAVs. To perform the auction process, the UAVs and the workers keep each other informed of their locations.

In our system model, we consider the use of UAVs to facilitate the FL training process by improving the communication efficiency between the FL workers and the FL model owner. At the same time, the UAVs can be used to support other types of users such as the typical mobile users, which generate background traffic for the UAVs. An extension to our system model to include different types of users, is straightforward and can be considered in the future work. Note that with greater computational and communication capabilities to support more IoV components for the training of FL models with higher accuracy level, the operational cost of involving the UAVs in the model training may be higher than that without using the UAVs.

\begin{table}[t]
\caption{Table of Commonly Used Notations.} 
\centering
\scriptsize
\begin{tabular}{|p{1.5cm} |p{6.4cm}| }

\hline
\textbf{Parameter}& \textbf{Description}\\ 
 \hline 

\rowcolor{lightgray}
\multicolumn{2}{|c|}{\emph{\textbf{Worker Parameters}}}\\[0.5ex]
\hline
$W_i$ & Total Number of Workers\\ \hline
$s_{iw}$ & Sampling Rate of Worker\\ \hline
$\sigma_{iw}$ & Importance of Worker\\ \hline

\rowcolor{lightgray}
\multicolumn{2}{|c|}{\emph{\textbf{Cell Parameters}}}\\[0.5ex]
\hline
$I$ & Total Number of Cells \\ \hline
$\sigma_i$ & Importance of Cell\\ \hline 
$q_i$ & Price Coefficient of Importance of Cell\\ \hline

\rowcolor{lightgray}
\multicolumn{2}{|c|}{\emph{\textbf{UAVs Parameters}}}\\[0.5ex]
\hline
$M$& Total Number of UAVs \\ \hline
$d^m_{i}$ & Distance between UAV and cell\\ [0.5ex] \hline
$E_m$ & Energy Capacity \\  \hline
$K_m$ & Cooperation Cost\\ \hline
$\delta_m$ & Price Coefficient of Energy Cost of UAV\\ \hline
$n_{im}$ & Maximum Number of Iterations of UAV \\ \hline

\rowcolor{lightgray}
\multicolumn{2}{|c|}{\emph{\textbf{Coalitions Parameters}}}\\[0.5ex]
\hline
$S_l$ & Coalition of UAVs\\ \hline
$\theta_1$, $\theta_2$ & Weight Parameters\\ \hline
$\alpha_i(S_l)$ & True Valuation \\ \hline
$x_i(S_l)$ & Profit of Coalition\\ \hline
$p_i(S_l)$ & Cost of Coalition \\ \hline
$c_i(S_l)$ & Payment Price of Bid Winner \\ \hline
$\gamma(\Pi)$ &Profit of Partition\\ \hline

\rowcolor{lightgray}
\multicolumn{2}{|c|}{\emph{\textbf{Model Owner Parameters}}}\\[0.5ex]
\hline
$\mu$ & Number of FL Iterations \\ \hline
\end{tabular}
\end{table}

\subsection{Worker Selection}
Each worker has different sensing capabilities and has different types of collected data. As such, the quality of data collected by each worker differs from one another. For example, a higher quality sensor provides data that more accurately reflects the changes in the road traffic environment. Reference \cite{Nguyen2013} considers the quality of data collected as a function of the sampling rate. The quality of data implies the importance of the workers where a worker with higher quality of data has higher importance, such that the model accuracy can be improved. Following a similar assumption as in \cite{Nguyen2013}, the importance of worker $w$ in cell $i$ is denoted by
\begin{equation}
\label{sampling}
\sigma_{iw}(s_{iw})=\frac{\log_{10}(s_{iw}+1)}{\log_{10}20},
\end{equation}
where $s_{iw}$ is the sampling rate of a specific worker denoted by $d_{iw}$. This concave function implies that the increase in sampling rate at lower values has a greater impact on the quality of worker as compared to the same increase in sampling rate at higher values. The log term in Equation (\ref{sampling}) captures the effect of diminishing marginal rate of return on the sampling rate as the data may be replicated or similar and thus, they are redundant \cite{xie2017anomaly}.

The model owner selects workers which has an importance level greater than $\zeta$, where $\zeta$ is the minimum threshold level of importance for the workers to be selected by the model owner. The value of $\zeta$ is determined by the model owner and is IoV application-specific. For example, $\zeta$ can be set larger for road safety applications such as navigation systems.

The importance of cell $i$ can be computed as a summation of individual importance of each selected worker in cell $i$. It can be expressed as:
\begin{equation}
\sigma_{i}=\sum_{w=1}^{W_i} \sigma_{iw},
\end{equation} 
where $\sigma_{iw}$ denotes the importance of worker $d_{iw}$.
 
\subsection{UAV Energy Model}

To execute the FL model training, each UAV incurs travelling cost to the workers' locations, communication and computational costs. It also incurs coalition formation cost that includes cooperation and coordination cost if a UAV decides to join a coalition. The energy spent by the UAV in executing the task includes the following components:
\begin{itemize}
\item flying from its original position to the destination and flying back to its original position,
\item receiving local model parameters from the workers,
\item staying at the destination to aggregate the local parameters from the workers,
\item transmitting aggregated local model parameters to the model owner, 
\item receiving global model parameters from the model owner,
\item transmitting global model parameters to the workers,
\item receiving aggregated local model parameters from the workers, and
\item communicating with other UAVs to form a coalition.
\end{itemize}

\subsubsection{\textbf{Flying Energy}}
Following a similar assumption as in \cite{YOO2016140}, the energy required by the UAV $m$ to fly from its depot to cell $i$ can be denoted by $e^f_{i,m}$, where
\begin{equation}
e^f_{i,m}=\left(\frac{1}{2}g^f_i+\frac{1}{2}g^f_m\right )e^f\frac{d^m_{i}}{v}.
\end{equation}

Without loss of generality, the UAVs are assumed to fly at a constant velocity $v$. The energy required to make a flight per unit time is denoted as $e^f$. We leverage $g^f_i$ and $g^f_m$ to denote the flying energy weight in cell $i$ and the depot of UAV $m$, respectively. The flying energy weights depend on the terrain of the cells and the depot of the UAV while $d^m_{i}$ reflects the distance between UAV $m$ and cell $i$. Note that the flying energy consumption, $e^f_{i,m}$ depends on both the locations of the UAV and the cell.

\subsubsection{\textbf{Computational Energy}}
Similar to \cite{latency2018}, the energy needed by UAV $m$ to complete the computation task per iteration, which is to aggregate the local model parameters of all workers, is represented as follows:
\begin{equation}
e^c_{m}=\kappa a_{m}f^2_{m},
\end{equation}
where $a_{m}$ is the total number of CPU cycles needed by UAV $m$ to complete the aggregation of local model parameters, $f_{m}$ is the computation capability of UAV $m$ which depends on the clock frequency of central processing unit (CPU) of UAV $m$, $\kappa$ is the coefficient of the value that depends on the circuit architecture of CPUs as in \cite{miettinen2010energy}.

\subsubsection{\textbf{Communication Energy}}

The UAV requires energy to transmit data packets to the model owner. As in \cite{chen2019joint}, the achievable uplink transmission rate by UAV $m$ can be represented by 
\begin{equation}
r^{UT}_m=\sum_{n=1}^{R}r_{mn}b_{m}\log_{2}(1+\frac{p_{mt}h_{m}}{I_{n}+b_{m}N_{0}}),
\end{equation}
where $R$ is the total number of resource blocks that the model owner allocates to the UAVs. Each UAV $m$ is only allowed to occupy one resource block. The term $r_{mn}$ is a resource block vector, where $r_{mn}=\{r_{m1},\ldots,r_{mn},\ldots,r_{mR}\}$. In particular, $r_{mn}=1$ indicates that resource block $n$ is allocated to UAV $m$, whereas $r_{mn}=0$ represents that the resource is not allocated to UAV $m$. The term $b_m$ is the bandwidth allocated to UAV $m$, $p_{mt}$ is the transmission power of UAV $m$, $h_m$ is the channel gain between UAV $m$ and the model owner, $I_n$ is the interference caused by other UAVs that uses the same resource block and $N_0$ is the power spectral density of the Gaussian noise. The issues of optimal resource block allocation, bandwidth allocation and transmission power allocation have been extensively studied in the literature \cite{greenallocation, qos}. Since it is not the focus of this paper, existing highly efficient schemes such as an opportunistic and efficient resource block allocation (OEA) algorithm \cite{greenallocation} can be applied to determine the allocation of resource blocks, bandwidth and transmission powers of the UAVs. The energy required by UAV $m$ to transmit model parameters in any cell $i$ to the model owner over wireless channels per iteration is

\begin{equation}
e^{UT}_{m}=p_{mt}\frac{z^c}{r^{UT}_m},
\end{equation}
where $z^c$ is the data size of cell-aggregated local model parameters. To upload the data packet of size $z^c$ over time $t$, the data size that UAV $m$ can handle must be larger than the input data size $z^c$.

The UAV also requires energy to receive the global FL parameters from the model owner. Similar to \cite{chen2019joint}, the achievable downlink transmission rate by the FL server to UAV $m$ is
\begin{equation}
r^{DR}_m=b_{BS}\log_{2}(1+\frac{h_{m}p_{BS}}{N_{0}b_{BS}}).
\end{equation}

The term of $b_{BS}$ is the bandwidth that the model owner uses to broadcast information to the UAVs and $p_{BS}$ is the transmission power of the model owner. The energy required by UAV $m$ to receive data from the model owner over wireless channels per iteration is
\begin{equation}
e^{DR}_{m}=p_{mr}\frac{z_{BS}}{r^{DR}_m},
\end{equation}
where $p_{mr}$ is the receiving power of UAV $m$. The term $z_{BS}$ is the size of the global model parameters.

Similarly, energy is also needed by the UAVs to communicate with the workers in the cell. The achievable uplink transmission rate by worker $w$ in cell $i$ is
\begin{equation}
r^{UR}_{iw}=\sum_{n'=1}^{R'_i}r_{iwn}b_{iw}\log_{2}(1+\frac{p_{iw}h_{iw}}{I_{n'}+b_{iw}N_{0}}),
\end{equation}
where $R'_i$ is the total number of resource blocks allocated to worker $w$ in cell $i$, $r_{iwn'}$ is a resource block vector where $r_{iwn'}=\{r_{iw1},\ldots,r_{iwn'},\ldots,r_{iwR'}\}$. The bandwidth allocated to and transmission power of worker $w$ in cell $i$ are denoted by $b_{iw}$ and $p_{iw}$, respectively. The term $h_{iw}$ is the channel  gain between the UAV and the worker and $I_{n'}$ is the interference caused by other workers using the same resource block.

The energy required by UAV $m$ to receive local model parameters from all the workers in cell $i$ per iteration is
\begin{equation}
e^{UR}_{im}=p_{mr}\sum_{w=1}^{W}\frac{z^w}{r^{UR}_{iw}},
\end{equation}
where $z^{w}$ is the data size of the local model parameters of the worker.

The achievable downlink transmission rate by UAV $m$ to the worker $w$ in cell $i$ is
\begin{equation}
r^{DT}_{iw}=b_{m}\log_{2}(1+\frac{h_{iw}p_{mt}}{N_{0}b_{m}}).
\end{equation}

The energy required by UAV $m$ to transmit the global model parameters to all workers in cell $i$ per iteration is
\begin{equation}
e^{DT}_{im}=p_{mt}\sum_{w=1}^{W}\frac{z_{BS}}{r^{DT}_{iw}}.
\end{equation}

\subsubsection{\textbf{Hovering Energy}}
Hovering energy is the energy needed by the UAV to stay near the cell. The hovering energy of UAV~$m$ in any cell $i$ is denoted by $e^{h}_{m}$.

\subsubsection{\textbf{Circuit Energy}}
Circuit energy is the energy consumed by the on-board circuits such as computational chips, rotors and gyroscopes \cite{circuit}. The on-board circuit energy is important and needs to be taken into consideration as it affects the energy-efficiency of the network. The circuit energy of UAV~$m$ is denoted as $e^{t}_{m}$.

In the next section, we present the coalition formation game among the UAVs.

\section{Coalitions of UAVs}
\label{sec:coalitions}

In this section, we present coalitions of UAVs. In particular, each cell $i$ charges the model owner for the local model parameters from the cells of workers. The amount of payment that the model owner needs to pay to the workers in cell $i$ for their contribution to train the FL model depends on the importance of cell $i$, which is $\sigma_i$. In other words, the workers in cell $i$ earn a revenue for training the FL model. The payment that cell $i$ receives for participating in the FL training process is $\frac{q_{i}\sigma_{i}}{t^c}$, where $q_i$ is the price paid by the model owner for a unit of importance in cell $i$ and $t^c$ is the total time needed to complete the FL task. Note that that the model owner only makes payment to the workers upon completion of the FL task. Since the payment is dependent on the total time needed for the completion of the FL task, the workers have higher incentive to complete the FL task as fast as they can. Hence, the workers have higher incentive to bid for nearer UAV coalitions which is further explained in Section~\ref{sec:auction}.

To receive full payment from the model owner, the cells need to complete $\mu$ iterations as required by the model owner. As a result, the cells of resource-constrained workers require the UAVs to facilitate the training process so as to improve communication efficiency. Since some UAVs may not have sufficient energy resources to support the entire FL process independently, the UAVs may form coalitions to support the cells of workers in completing the FL training.

\subsection{Coalition Game Formulation}
We have the following definitions for coalition game formulation. 

\begin{definition}
A coalition of UAVs is denoted by $S_{l}\subseteq \mathcal{M}$, where $l$ is the index of the coalition \cite{coalition}.
\end{definition}

\begin{definition}
A partition or coalitional structure is a group of coalitions that spans all players in $\mathcal{M}$ and is represented as $\Pi=\{S_1,\ldots,S_l,\ldots,S_L\}$, where $S_{l}\cap S_{l'}=\emptyset$ for $l\neq l'$, $\bigcup_{l=1}^{L}S_{l}=\mathcal{M}$ and $L$ is the total number of coalitions in partition $\Pi$. 
\end{definition}

The total number of possible partitions for $M$ players is given by $D_M$, which is known as the Bellman number given as follows:
\begin{equation}
\label{bellman}
D_u=\sum_{j=0}^{u-1} {{u-1}\choose{j}} D_{j}, \; \text{for} \; u\geq1\; \text{and} \; D_0=1,
\end{equation}
$\mathcal{M}$ denotes the coalition of all players, which is also known as the grand coalition.
\begin{definition}
A partition $\Pi=\{S_1,\ldots,S_l,\ldots,S_L\}$  is a stable partition if no coalition $S_l$ has incentive to change the current partition $\Pi$ by joining another coalition $S_{l'}$, where $S_{l}\cap S_{l'}=\emptyset$ for $l\neq l'$, or splitting into smaller disjoint coalitions \cite{coalition}.
\end{definition}

The term of $K_m(S_l)$ is the cooperation cost incurred by UAV $m$ to form a coalition $S_l$. 
It includes the cost of communicating with other UAVs in the same coalition. The cooperation cost is only incurred when there are two or more UAVs in a coalition. Otherwise, there is no cooperation cost incurred by the UAV. Specifically, $K_m$ is defined as:
\begin{equation}
\label{cooperatecost}
  K_m(S_l)=\begin{cases}
    K_m(S_l), & \text{if $|S_l| \geqslant 2$}\\
    0, & \text{otherwise}
  \end{cases}, \; \forall S_l \in \Pi,
\end{equation}
where $|S_l|$ is number of coalition members in the coalition $S_l$. The coalition formation cost is a non-decreasing function of the size of the coalition. Note that there is no communication between the UAV coalitions.

When two or more UAVs cooperate to form a coalition $S_l$ to support the FL training process in cell $i$, the FL training process starts as soon as the nearest UAV reaches cell $i$. This nearest UAV aggregates the local model parameters which are uploaded by the workers in cell $i$, and transmits the aggregated local model parameters to the model owner for the maximum number of iterations that it can support given its energy capacity. Then, the second nearest UAV in the same coalition takes over, aggregates the local model parameters and transmits the aggregated local model parameters to the model owner. Similarly, after the second nearest UAV completes its task, the third nearest UAV takes over and continues the process. If the maximum number of iterations that the UAV can support is greater than the remaining number of iterations of the FL training process, the UAV just needs to support the remaining number of iterations. If the required number of iterations announced by the model owner is completed by the nearer UAVs in a coalition, the farther UAVs do not need to support any iteration of the FL training process. For fairness purposes, since each UAV in the coalition completes different number of iterations, the profit earned by each UAV in the coalition depends on the number of iterations it completes.

Given that each UAV $m$ has an energy capacity of $E_m$, the maximum number of iterations that each UAV $m$ in coalition $S_l$ can support to facilitate the FL training process in cell $i$ is denoted by $n_{im}(S_l)$, where
\begin{equation}
\label{number}
n_{im}(S_l)=\frac{E_{m}-2e^f_{i,m}-K_{m}(S_l)}{e^{UR}_{im}+e^{DR}_{m}+e^{UT}_{m}+e^{DT}_{im}+e^{c}_{m}+e^{h}_{m}+e^{t}_{m}}.
\end{equation}

Since each coalition $S_l$, $\forall S_l \in \Pi$ receives different revenues for facilitating the FL training process in different cell $i$, an auction scheme is needed to elicit the valuations of different cells for each coalition in order to decide on the optimal allocation of UAV coalitions to the cells of workers given a partition $\Pi$. The mechanism of the auction scheme is explained further in detail in the next section. Note that the total number of iterations that all UAVs in a coalition $S_l$, $\forall S_l \in \Pi$ can support determines the choices of coalitions that each cell of workers can bid for. The coalition formation game determines the stable partition that allows the UAVs to earn a maximum profit.

\begin{algorithm}[t]
\caption{Algorithm for Coalition Formation of UAVs using Merge-and-Split.}
\footnotesize
\label{mergesplit}
\begin{algorithmic}[1]
 \renewcommand{\algorithmicrequire}{\textbf{Input:}}
 \renewcommand{\algorithmicensure}{\textbf{Output:}}
 \REQUIRE Set of UAVs $\mathcal{M}=\{1,\ldots,m,\ldots,M\}$
  \ENSURE Partition that provides highest profit $\Pi^*=\{S^*_1,\ldots, S^*_y,\ldots,S^*_Y\}$
 
 \STATE Initialize a set of partitions $\mathcal{Z} = \{\Pi_1,\ldots, \Pi_t,\ldots,\Pi_T\}$ that includes all possible partitions where $T$ is the total number of partitions

 \STATE Compute allocations of UAVs to cells of workers and the maximum obtainable profit given current partition, $\Pi_{curr}$ from Algorithm \ref{allocation}
 
 \STATE \textbf{\emph{\underline{Merge Mechanism:}}}
 	\FOR {\textbf{each} coalition $S_l$ in $\Pi_{curr}$}
 	\STATE Consider hypothetical partition $\Pi_{new}$ where $S_l \cup S_{l'}$, $l\neq l'$
		\STATE Compute the optimal allocation of UAVs and maximum attainable profit, $\gamma(\Pi_{new})$ from Algorithm \ref{allocation}
			\IF {$\gamma(\Pi_{new})>\gamma(\Pi_{curr})$}
				\STATE Merge coalitions $S_l$ and $S_{l'}$
				\STATE Update partition $\Pi_{curr} \gets \Pi_{new}$
				\STATE Update total profit $\gamma(\Pi_{curr}) \gets \gamma(\Pi_{new})$
			\ENDIF
	\ENDFOR
	
\RETURN $\Pi_{merge}=\{S_1,\ldots, S_g,\ldots,S_G\}$
 
\STATE \textbf{\emph{\underline{Split Mechanism:}}}
	\FOR {\textbf{each} coalition $S_g$ in $\Pi_{merge}$}
		\STATE Initialize set of possible coalition splits $\tilde{S}_g =\{\tilde{S}_1,\ldots,\tilde{S}_g,\ldots,\tilde{S}_G\}$\FOR {\textbf{each} coalition $\tilde{S}_g$ in $S_g$}
				\STATE Consider hypothetical partition $\Pi'_{new}$
				\STATE Compute the optimal allocation of UAVs and maximum attainable profit, $\gamma(\Pi'_{new})$ from Algorithm \ref{allocation}
					\IF {$\gamma(\Pi'_{new})>\gamma(\Pi_{curr})$}
						\STATE Split coalitions $S_g$ such that $S_g \gets \tilde{S}_g$
						\STATE Update partition $\Pi_{curr} \gets \Pi'_{new}$
						\STATE Update total profit $\gamma(\Pi_{curr}) \gets \gamma(\Pi'_{new})$
					\ENDIF
			\ENDFOR
	\ENDFOR

\RETURN Final partition that provides the highest profit $\Pi^*=\{S^*_1,\ldots, S^*_y,\ldots,S^*_Y\}$
\end{algorithmic}
\end{algorithm}

\begin{proposition}
The formation of grand coalition, where all the UAVs join a single coalition, is not always stable.
\end{proposition}


\begin{proof}
In order to prove that the grand coalition is not always the optimal coalitional structure among the UAVs, we show that the coalitional game is not superadditive and prove that the core of the proposed game is empty by following the procedure in \cite{mergeandsplit}.

A coalitional game is superadditive if the cooperation among disjoint coalitions to form a larger coalition, guarantees a payoff which is at least equal to that obtained by the disjoint coalitions separately, i.e., $x(S_1 \cup S_2)\geq x(S_1)+x(S_2), S_1\in\mathcal{M}, S_2\in\mathcal{M}$ and $S_{1}\cap S_{2}=\emptyset$. Consider two coalitions where coalition $S_1$ is allocated to cell $i$ to support the FL task in cell $i$ and earning a profit of $x_i(S_1)$ whereas coalition~$S_2$ is not allocated to any cell and earning zero profit, i.e., $x_i(S_2)=0, \; \forall i \in \mathcal{I}$ . By merging coalition $S_1$ and coalition~$S_2$ to support the same cell $i$, the revenue, i.e., the payment price of cell $i$ does not change, but the total cost in terms of energy incurred and coalitional cost also increases. Thus, $x_i(S_1 \cup S_2)< x_i(S_1)$, i.e., the marginal profit of the UAVs from merging the two coalitions is negative. Therefore, the coalitional game is non-superadditive.

Next, we prove that the core of the proposed game is empty. As defined in \cite{mergeandsplit}, an imputation is a payoff vector $\mathbf{z}=\{z_1,\ldots,z_m,\ldots,z_M\}$ that satisfies the following two conditions:
\begin{enumerate}
    \item group rational, i.e., $\sum_{m \in \mathcal{M}}z_m=x(\mathcal{M})$, and 
    \item individually rational where each player obtains a benefit no less than acting alone, i.e., $z_m\geq x(\{m\})\;\forall i$.
\end{enumerate}
The core of the coalition is the set of stable imputations where there is no incentive for any coalition $S_l \in \mathcal{M}$ to reject the proposed payoff allocation $\mathbf{z}$, deviate from the grand coalition and form coalition $S_l$. The core of the coalition is defined as:
$$\tau=\Bigg\{\mathbf{z}\colon\sum_{m \in \mathcal{M}}z_m=x(\mathcal{M}) \; \mathrm{and} \sum_{m\in S_l}z_m\geq x(S_l)\; \forall S_l \in \mathcal{M} \Bigg\}$$
As we have previously established that the marginal profit of merging the coalitions can be negative, this means that the condition of individual rationality is violated, i.e., $\sum_{m\in S_l}z_m< x(S_l)$. In particular, the profit is higher if coalition $S_1$ does not merge with coalition $S_2$ to form a larger coalition. The core of the proposed coalitional game is empty.

Therefore, due to the non-superadditivity of the coalitional game and the emptiness of the core, the grand coalition does not form among the UAVs. Instead, smaller and disjoint coalitions of UAVs will form in the FL network.
\end{proof}

As such, we discuss the joint auction-coalition formation algorithm next.

\subsection{Coalition Formation Algorithm}
\label{subsec:coalitionalgorithm}

Following \cite{apt2006stable}, we define two simple merge-and-split operations to modify a partition $\Pi$ in $\mathcal{M}$.
\begin{itemize}
\item \textbf{Merge Rule:} Merge any set of coalitions $\{S_1,\ldots,S_l\}$  where $\gamma(\Pi) <  \gamma(\Pi_{new})$, thus $\{S_1,\ldots,S_l\} \to \bigcup^{l}_{j=1}S_j$. 

\item \textbf{Split Rule:} Split any set of coalitions $\{\bigcup^{l}_{j=1}S_j\} $ where  $\gamma(\Pi) <  \gamma(\Pi'_{new})$, thus $\bigcup^{l}_{j=1}S_j\ \to \{S_1,\ldots,S_l\}$.

\end{itemize}

The merge-and-split algorithm for the coalition formation is presented in Algorithm \ref{mergesplit}.

The merge mechanism is illustrated from the perspective of a representative coalition $S_l$ in a given partition $\Pi$. The coalition $S_l$ decides to merge with another coalition $S_{l'}$ where $l \neq l'$ and $S_l \in \mathcal{M}$ if the total profit of the UAVs is improved. In order to evaluate the profits, the UAVs assume the hypothetical partition $\Pi_{new}$ and compute maximum attainable profit of the hypothetical partition, $\gamma(\Pi_{new})$ (lines 5-6), of which the bidding of the workers and the allocation of UAVs are discussed in the next section. If the total profit of the UAVs in the hypothetical partition $\Pi_{new}$, which is denoted by $\gamma(\Pi_{new})$, is higher than that in its current partition $\Pi_{curr}$, which is represented as $\gamma(\Pi_{curr})$, coalition $S_l$ will merge with coalition $S_{l'}$ to form new partition $\Pi_{new}$ (lines 7-8). Then, given this new partition $\Pi_{new}$, the current partition and the value of the maximum profit will be updated (lines 9-10). On the other hand, $\gamma(\Pi_{new})$ is less than that of the current partition $\Pi_{curr}$, $\gamma(\Pi_{curr})$, there will be no change to the current partition $\Pi_{curr}$. For the next iteration, any coalition $S_l$ in the given partition, $\Pi_{curr}$, will consider to form a coalition with another coalition $S_{l'}$. The merge mechanism terminates when all possible partitions have been considered. At the end of the merge mechanism, the algorithm returns a resulting partition $\Pi_{merge}=\{S_1, \ldots,S_g,\ldots, S_G\}$ (line 13). Consequently, the split mechanism is carried out.

Similarly, the split mechanism is illustrated from the perspective of a representative coalition $S_g$ in partition $\Pi_{merge}$. Let $\mathcal{\tilde{S}}_g$ be the set of possible coalition splits for coalition~$S_g$, where $\mathcal{\tilde{S}}_g =\{\tilde{S}_1,\ldots,\tilde{S}_g,\ldots,\tilde{S}_G\}$ (line 16). Coalition~$S_g$ considers one of the possible ways to split its coalition into smaller disjoint coalitions (line 17). Similar to the merge mechanism, in order to evaluate the profits, the UAVs assume the hypothetical partition $\Pi'_{new}$ and retrieve the bids from the auction and compute their maximum obtainable profits (lines 18-19). If the total profit of the UAVs is improved with more disjoint coalitions, the coalition $S_g$ will proceed to split and form a new partition $\Pi'_{new}$ (lines 20-21). The current partition and the value of maximum profit will be updated (lines 22-23). Then, given the new partition $\Pi'_{new}$, any coalition will split and evaluate their profits again. The partition is maintained if the profit is higher than if the split is carried out. The split mechanism terminates when all possible partitions with smaller coalitions have been considered. At the end of the merge-and-split algorithm, the final partition $\Pi^*=\{S^*_1,\ldots,S^*_y,\ldots,S^*_Y\}$ is returned (line 27).

\begin{proposition}
The final partition $\Pi^*=\{S^*_1,\ldots,S^*_y,\ldots, S^*_Y\}$ is a stable partition that maximizes the total profit of the UAVs.
\end{proposition}

\begin{proof}
Given any current partition $\Pi_{curr}=\{S_1,\ldots,S_l,\ldots, S_L\}$, where $S_{l}\cap S_{l'}=\emptyset$ for $l\neq l'$, $\bigcup_{l=1}^{L}S_{l}=\mathcal{M}$, the current partition $\Pi_{curr}$ is only updated when the total profit of the hypothetical partition $\Pi_{new}$ is higher than that of the current partition $\Pi_{curr}$ by either merging two disjoint coalitions or splitting a single coalition into multiple smaller disjoint coalitions. Therefore, the merge-and-split algorithm (Algorithm \ref{mergesplit}) generates a sequence of partitions where each partition has either equal profit as or higher profit than the partition from previous iteration. The total number of possible partition $D_M$ is a finite number. The algorithm will eventually terminate at a partition $\Pi^*=\{S^*_1,\ldots,S^*_y,\ldots, S^*_Y\}$ where the total profit of partition $\Pi^*$ cannot be increased further by merging or splitting coalitions. Hence, the final partition $\Pi^*=\{S^*_1,\ldots,S^*_y,\ldots, S^*_Y\}$ is stable. In addition, at stable partition $\Pi^*=\{S^*_1,\ldots,S^*_y,\ldots, S^*_Y\}$, since the total profit of the UAVs cannot be increased further by changing the partition $\Pi^*$, the total profit given by partition $\Pi^*=\{S^*_1,\ldots,S^*_y,\ldots, S^*_Y\}$ is the maximum.
\end{proof}
  
Next, we present an auction design to elicit the valuations of cells of workers for all possible UAV coalitions in order to determine the allocation of UAV coalitions to the cells of workers, thereby maximizing the profit of the UAVs given a partition $\Pi$.

\section{Auction Design}
\label{sec:auction}

We consider multiple auctioneers and multiple buyers at the same time. The UAVs are the auctioneers which conduct the auctions and offer the resources to facilitate the FL training process. The cells of workers are the buyers that submit bids and pay the UAVs for their energy resource. 

\subsection{Buyers' Bids}
Firstly, the workers in cell $i$ ($\forall i \in \mathcal{I}$) submit their bids to each possible coalition $S_l$ in a given partition $\Pi$ to indicate their valuations for each coalition. 

Each cell needs to build their own preference over all the possible coalitions that can provide services, where cells of workers need to compare the coalitions and rank their preferences. As such, the concept of  preference relation is introduced to evaluate the order of preference of each player.

\begin{definition}
For any player $i \in \mathcal{I}$, a preference relation \cite{hedonic} $\succeq_i$ is defined as a complete, reflexive, and transitive binary relation over the set of all coalitions that cell $i$ can choose. 
\end{definition}

In particular, for any cell $i \in \mathcal{I}$, given $S_1 \subseteq \mathcal{M}$ and $S_2 \subseteq \mathcal{M}$, 
$S_1 \succeq_i S_2$ means that cell $i$ prefers coalition $S_1$ over coalition $S_2$, or at least cell $i$ prefers both coalition equally. The asymmetric counter part of $\succeq_i$, which is represented by $\succ_i$, when used in $S_1 \succ_i S_2$, means that cell $i$ strictly prefers coalition $S_1$ over coalition $S_2$. It is worth noting that the preference relation is defined to quantify the preferences of the players, which can be application-specific. For example, the IoV components prefer UAVs that are nearer for road safety applications but prefer UAVs that provide higher bandwidth for entertainment services. It can represent a function of parameters such as the utility gain by joining the coalition.

The preference relation for any cell $i \in \mathcal{I}$ can be represented by the following formula,
\begin{equation}
\label{relation}
S_1 \succeq_i S_2 \iff \alpha_i(S_1)\geq \alpha_i(S_2),
\end{equation}
where $S_1\subseteq \mathcal{M}$ and $S_2 \subseteq \mathcal{M}$, 
$\alpha_i(S_l)$ is the preference function for any cell $i \in \mathcal{I}$ and for any coalition $S_l$ facilitating FL training process in cell $i$.

In general, the more important cells of workers have incentives to pay the UAVs more to support the FL training process as they receive higher reward for completing the FL task. Meanwhile, the valuations of the cells for the UAVs depend on the latency of the task, i.e., how quickly the model parameters are transmitted to the model owner. As such, the cells of workers have more incentive to employ UAVs that are nearer to them since shorter time is required for them to reach the cell to facilitate FL training. Although the FL training process starts as soon as the nearest UAV in a coalition reaches the cell of workers, the valuation of the cell of workers for the UAV coalition does not depend on the nearest UAV in the coalition. In fact, the valuation of the cell of workers for a coalition $S_l$ depends on the farthest UAV in the coalition as they need to wait for the farthest UAV to arrive to continue and complete the FL training process, even if other nearer UAVs can reach the cell and support their parts of the training process in a short amount of time. Specifically, the time for the farthest UAV to reach the cell may be longer than the time needed by the nearer UAVs to travel to the cell and complete their parts of the FL training process.

Thus the valuation of workers in cell $i$ for coalition $S_l$, which is also the preference function, is defined as
\begin{equation}
\alpha_{i}(S_l)=\theta_{1}{q_{i}\sigma_{i}}+ \frac{\theta_2}{\max \limits_{m \in S_l}{t^m_{i}}},
\end{equation}
where $t^m_i$ is the time needed for UAV $m$ to travel to cell $i$, $\theta_1$ is the weight parameter of the importance of the cell of workers whereas $\theta_2$ is the weight parameter of the time needed for the UAVs to reach the cells. In particular, each cell $i$ submits its valuations for each UAV coalition in a given partition $\Pi$. Each cell $i$ will only submit its bids to any coalition $S_l$, $\forall S_l \in \Pi$  if the total number of iterations that can be completed by the coalition $S_l$ is greater than the required number of iterations specified by the model owner. In other words, coalition $S_l$ will only receive bids from any cell $i$ if the following condition is satisfied:
\begin{equation}
\sum_{m\in S_l}n_{im}(S_l)\geq \mu. 
\end{equation}

The traditional first-price auction, in which the highest bidder pays the exact bid it submits, maximizes the revenue of the seller, but does not guarantee that the bidders submit their true valuations. In other words, the bidders have incentives to not submit their true valuations. In particular, if the bidder bids lower and wins the bid, its utility increases. Hence, by adopting the first-price auction, the property of incentive compatibility of the auction is not guaranteed. As such, the second-price auction is adopted to determine the payment price of the bid winners where the bid winners need to pay only the price equals the winning bid if the original bid winner is not considered in the auction. The payment price is determined to incentivize the cells of workers to honestly report their true valuations such that the valuations of the cells of workers are equal to their bids, i.e., $\alpha_i(S_l)=\lambda_i(S_l)$, where $\lambda_i(S_l)$ is the bid submitted by cell $i$ to coalition $S_l$. The payment price of cell $i$ for coalition $S_l$ is represented by $p_i(S_l)$. Thus, the utility of cell $i$ if it wins the bid is defined as:
\begin{equation}
u_i(S_l)=\alpha_i(S_l)-p_i(S_l). 
\end{equation}
On the other hand, if cell $i$ is not allocated any UAV coalition, the utility is $u_i(S_l)=0$.

\subsection{Sellers' Problem}

The UAVs as the auctioneers need to decide on the optimal allocation to different cells of workers. The algorithm for the allocation of profit-maximizing UAVs is presented in Algorithm \ref{allocation}.

The objective of the UAVs is to maximize sum of their individual profits where the UAVs cooperate fully with each other and are not selfish. In order to calculate the maximum total profits for any partition $\Pi$, the maximum possible profit for each coalition in the partition needs to be calculated first. The coalition will only earn the actual profit if it is allocated to support the FL task in any cell $i$. Given any partition $\Pi$, the revenue of coalition $S_l$ in supporting cell $i$ is the payment price, i.e., $p_{i}(S_l)$, which can be evaluated from the bid values (lines 3-4). The cost of coalition $S_l$ in supporting cell $i$ is the summation of energy incurred by each UAV in the coalition (line 5). The cost of coalition $S_l$ in supporting cell $i$ is denoted as $c_i(S_l)$, which is defined as:
\begin{multline}
c_i(S_l)=\delta_m\sum_{m \in S_l}2e^f_{i,m}+\sum_{m \in S_l}K_m(S_l)
\\+\delta_m\sum_{m \in S_l} n^c_{im}(e^{UR}_{im}+e^{DR}_m+e^{UT}_m+e^{DT}_{im}+e^c_m+e^h_m+e^t_m),
\end{multline}
where $n^c_{im}$ is the number of iterations that is supported by UAV $m$ in cell $i$ and $\delta_m$ is the cost coefficient per unit of energy spent by UAV $m$. Note that it is not necessary that each UAV in the coalition supports the maximum number of iterations, $n_{im}$. If the number of remaining iterations is less than the maximum number of iterations that the UAV can support, the UAV needs to only support the remaining number of iterations, which requires less energy.

The possible profit earned by coalition $S_{l}$ for supporting FL training process in any cells $i \in \mathcal{I}$ is the difference between the revenue earned and the cost incurred by coalition $S_l$ in facilitating the FL training in cell $i$ (line 5). The possible profit denoted by $x_i(S_l)$ can be represented by
\begin{equation}
x_i(S_l)=p_i(S_l)-c_i(S_l). 
\end{equation}

In this regard, each coalition $S_{l}$ in partition $\Pi$ can now evaluate the most preferable cell that provides them with the maximum profit (line 12). If two or more coalitions choose the same cell of workers, the coalition of UAVs which is the most preferred based on the bids submitted will be allocated (lines 13-14). For example, according to the preference relation in Equation (\ref{relation}), coalition $S_l$ will be allocated to cell $i$, instead of coalition $S_{l'}$ if $\alpha_i(S_l)\geq \alpha_i(S_{l'})$. The coalition of UAVs which are not allocated any cell will not earn any profit. If there is no competition among coalitions of UAVs for a cell, the UAV coalition is allocated to the cell (lines 15-16). Given partition~$\Pi$ and the allocation of UAV coalitions to the cells, the total profits can be calculated by summing up the individual profits of each coalition which is allocated a cell (lines 18-19). The total profits of partition $\Pi$ is given by $\gamma(\Pi)$, which is defined as follows:
\begin{equation}
\gamma(\Pi)=\sum^{I}_{i=1}\sum^{L}_{l=1}\beta_{il} x_i(S_l),
\end{equation}
where $\beta_{il} = \{\beta_{i1},\ldots,\beta_{il},\ldots,\beta_{iL}\}$ is an allocation vector which indicates whether coalition $S_l$ is allocated to cell $i$. In particular, $\beta_{il}=1$ if coalition $S_l$ is allocated to cell $i$ and $\beta_{il}=0$ if coalition $S_l$ is not allocated to
cell $i$.

Then, the allocated coalition is removed from the current partition, $\Pi$ and the corresponding cell is also removed from the list of cells (line 20). The number of cells to be supported reduces by one (line 21).

\begin{algorithm}[t]
\caption{Algorithm for Allocation of Profit-Maximizing UAVs.}
\footnotesize
\label{allocation}
\begin{algorithmic}[1]
 \renewcommand{\algorithmicrequire}{\textbf{Input:}}
 \renewcommand{\algorithmicensure}{\textbf{Output:}}
 \REQUIRE Current partition $\Pi=\{S_1,\ldots,S_l,\ldots,S_L\}$
 \ENSURE Maximum obtainable profit, $\gamma(\Pi)$ and allocation of UAVs to cells of workers
	\FOR {\textbf{each} cell of workers $i$ in $\mathcal{I}$}
    			\FOR {\textbf{each} coalition of UAVs $S_l$ in $\Pi$}
    				\STATE Submit bid values, $\alpha_i(S_l)$
    				\STATE Compute payment price, $p_i(S_l)$ if cell $i$ is the bid winner for coalition $S_l$
				\STATE Compute the cost of each UAV coalition $S_l$ in supporting cell $i$
				\STATE Compute the profit of each UAV coalition $S_l$ in supporting cell $i$, $x_{i}(S_l)$
    			\ENDFOR
    	\ENDFOR
	\STATE Initialize $t=I$
	\STATE Initialize $\gamma(\Pi)=0$
	\WHILE{$t \neq 0$}
		\STATE Identify the highest profit for each coalition and its corresponding cell, $x_i(S_l)$ in $\Pi$ 
		\IF {Two or more UAV coalitions compete for the same cell}
			\STATE The cell with the highest bid is allocated the coalition
		\ELSE
			\STATE The coalition (without competition) is allocated the cell
		\ENDIF
		
		\STATE Update partition $\Pi$
		\STATE $\gamma(\Pi)\gets \gamma(\Pi)+x_i(S_l)$
		\STATE Remove allocated UAVs and cell from the partition $\Pi$
		\STATE $t=t-1$
	\ENDWHILE
\RETURN $\gamma(\Pi)$ and allocation of UAVs to cells of workers
\end{algorithmic}
\end{algorithm}

\subsection{Analysis of the Auction}

We now analyze the individual rationality and the incentive compatibility of the auction mechanism. Firstly, to encourage the buyers to participate in the auction, the auctioneers need to ensure that the buyers receive positive payoffs. Moreover, the auction mechanism discourages the buyers from submitting untruthful valuations that do not reflect the true valuations of buyers' bids.

\begin{proposition}
Each buyer, i.e., cell of workers, achieves individual rationality in this auction mechanism.
\end{proposition}
\begin{proof}
Firstly, for each cell of workers which has no winning bid, its utility is zero, i.e., $u_i(S_l)=0$ since it does not earn a revenue from the model owner and it also does not pay any UAV coalition to complete the FL task. Secondly, for the cells of workers with winning bids, i.e., allocated a UAV coalition to support the FL training process, their utility is:
\begin{equation}
u_i(S_l)=\alpha_i(S_l)-p_i(S_l)
=\lambda_i(S_l)-p_i(S_l),
\end{equation}
The utility gained by the bid winner is non-negative, i.e., $u_i(S_l) \geqslant 0$. When $\lambda_i(S_l)$ wins the bid, its true valuation must be greater than the payment price, i.e., $\alpha_i(S_l) \geqslant p_i(S_l)=\lambda_{i'}(S_l)$. Otherwise, $\lambda_{i'}(S_l)$ wins the bid instead.
\end{proof}

In order to prove the property of incentive compatibility, i.e., the truthfulness of buyers' bids, we show that there is no incentive for the buyers to submit bids other than their true valuations by following the procedure in \cite{winningbid}. The auction mechanism is truthful if it satisfies the following two conditions:
\begin{enumerate}
\item The winning bid algorithm is monotonic.
\item The pricing is independent of the winning bid.
\end{enumerate}

\begin{proposition}
\label{monotonic}
The winning bid algorithm is monotonic. For each bid $\lambda_i(S_l)$, if bid $\lambda_i(S_l)$ wins, then bid $\lambda'_{i}(S_l)$ also wins, where $\lambda'_i(S_l)=\lambda_i(S_l)+\phi$ and $\phi > 0$.
\end{proposition}
\begin{proof}
Suppose that $\lambda_i(S_l)$ is the winning bid with the UAV coalition $S_l$ earning a profit $x_i(S_l)$, a higher bid $\lambda'_i(S_l)$, where $\lambda'_i(S_l)= \lambda_i(S_l)+\phi$ and $\phi>0$, will result in the same profit $x_i(S_l)$, since the cost of coalition $S_l$ for supporting the FL training process in cell $i$ and the price of the winning bid, $p_i(S_l)$ are the same, $\lambda'_i(S_l)$ still wins the bid.
\end{proof}
\begin{proposition}
\label{sameprice}
If cell $i$ wins the bidding with $\lambda_i(S_l)$ and $\lambda'_i(S_l)$, the payment price $p_i(S_l)$ is the same for both bids.
\end{proposition}
\begin{proof}
As long as cell $i$ wins the bid, the payment price of the bid winner is the second price. By bidding either $\lambda_i(S_l)$ or $\lambda'_i(S_l)$, the payment price does not change. Hence, the payment price is independent of the winning bid $\lambda_i(S_l)$ or $\lambda'_i(S_l)$.
\end{proof}

\begin{proposition}

No buyer, i.e., cell of workers can achieve higher payoff by bidding with values other than its true valuations, i.e., if a buyer bids $\lambda'_i(S_l)\neq \lambda_i(S_l)$, its utility is $u'_i(S_l) \leq u_i(S_l)$.
\end{proposition}
\begin{proof}
Following a similar procedure in \cite{winningbid}, we consider two cases.

\emph{Case 1.} If cell $i$ wins the bid for coalition $S_l$, it needs to pay $p_i(S_l)$ which is equal to the second highest winning bid. If cell $i$ decides to bid untruthfully with $\lambda'_i(S_l)\geq \lambda_i(S_l)$, cell~$i$ still wins the bid according to Proposition \ref{monotonic} without any increase in utility. On the other hand, if cell $i$ decides to bid untruthfully with $\lambda'_i(S_l)\leq \lambda_i(S_l)$, cell $i$ may still win the bid without any increase in utility according to Proposition \ref{sameprice} or lose the bid, which leads to a reduction in utility.

\emph{Case 2.} If cell $i$ loses the bid for coalition $S_l$, it does not need to pay anything. If cell $i$ bids untruthfully with $\lambda'_i(S_l)\geq \lambda_i(S_l)$, cell~$i$ may win the bid but suffer from a non-positive utility or lose the bid with no change in utility. If cell $i$ bids untruthfully with $\lambda_i(S_l)\leq \lambda_i(S_l)$, it will still lose the bid.
Therefore, the buyer is not incentivized to bid untruthfully since by doing so, it is not able to achieve higher utility.
\end{proof} 

Therefore, the auction mechanism has the properties of individual rationality and incentive compatibility.

\subsection{Complexity of the Joint Auction-Coalition Algorithm}

In this subsection, we investigate the complexity of the joint auction-coalition algorithm. The time complexity of the proposed algorithm depends on the number of attempts of the merge and split operations. Since for each attempt involves the auction-based allocation of which the time complexity increases linearly with the number of cells of workers, the overall time complexity of the proposed algorithm is given by the multiplication of the number of merge and split attempts and the complexity of the auction-based allocation.

For a given partition, the merge operation is carried out if the merging of two coalitions results in a higher total profit than if the two coalitions operate independently. In the worst case scenario, each UAV coalition attempts to merge with all other UAV coalitions. At the start, all UAVs form singleton coalitions, which result in $M$ coalitions. In the worst case, $\frac{M(M-1)}{2}$ number of attempts required before the first merge operation is carried out, $\frac{(M-1)(M-2)}{2}$ number of attempts required before the second merge operation is carried out and so on. The total number of attempts of merge operations is in $O(M^3)$ for the worst case scenario. However, in practice, the merge process requires a significantly lower number of attempts. Once the UAV coalitions merge to form a larger coalition, the number of merge attempts per coalition reduces since the number of merging possibilities for the remaining UAVs decreases.

In the worst case scenario, splitting a UAV coalition $S_l$ involves finding all possible partitions of the UAVs in the coalition. The number of possible partitions of a set of UAVs is given by the Bellman number in Equation (\ref{bellman}), which increases exponentially in the number of UAVs in the set. In practice, the split operation is not performed over the set that consists of all UAVs, but over the smaller set of each formed UAV coalition. Since the UAVs incur higher cooperation cost by forming a larger coalitions in Equation (\ref{cooperatecost}), they prefer to form coalitions of smaller size where possible. As a result, the split operation is affordable in terms of complexity.

\section{Simulation Results and Analysis}
\label{sec:sim}

\begin{figure}
\includegraphics[width=\linewidth]{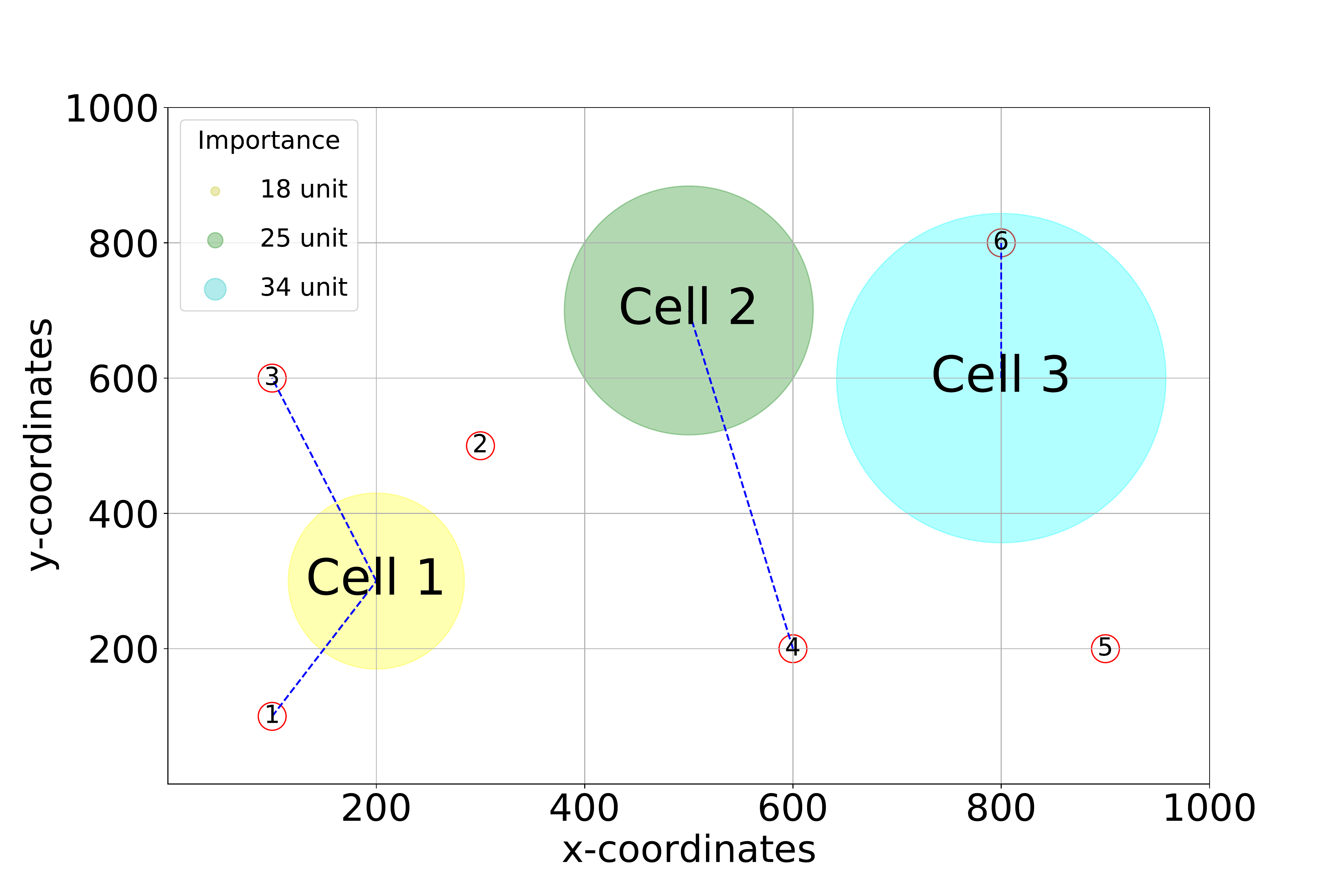}
\caption{Distributed FL Network with 3 cells and 6 UAVs.}
\label{cell}
\end{figure}

In this section, we evaluate the joint auction-coalition formation framework of the UAVs. Firstly, we evaluate communication efficiency in the UAV-enabled network. Then, we analyze the preference of each cell for each UAV, followed by the profit-maximizing behavior of the UAVs and the allocation of UAVs to different cells of workers. Finally, we compared the proposed framework against existing schemes. We consider a distributed FL network on a Cartesian grid of size 1000$\times$1000 as shown in Fig. \ref{cell}. All simulation parameters are summarized in Table~\ref{simulated}.

The sensors of the IoV components such as the vehicles and the RSUs typically have sampling rates that range between 250 samples per second (sps) to 32000 sps. The IoV components, i.e., the FL workers are selected to train the FL model if their importance, $\sigma_{iw}\geq 1$. The training datasets that the FL workers use to train the FL model consist of the samples collected by their sensors. Thus, the FL workers with sensors of higher sampling rates have larger training datasets as compared to the FL workers with sensors of lower sampling rates. Note that since the selected FL workers only transmit the local model parameters to the model owner, the size of the training datasets does not affect the size of the transmitted local model parameters.

\subsection{Communication Efficiency in FL Network}

The introduction of UAVs to the FL network can improve the communication efficiency. According to \cite{chen2019joint}, we set the bandwidth of IoV components, denoted by $b_{iw} = $ 150kHz, the channel gain between FL server and the IoV components, $h_{iw} = $  2dBm, $N_0 =$  -174 dBm/Hz and the range of transmission power, $p_{iw}$ is between 1mW and 10mW. For the UAVs, similar to \cite{uavpower}, we set the transmission power of UAVs, $p_{mt}$, to be between 0.5W and 5W. We also set the bandwidth of the UAVs, $b_m =$ 400kHz and the channel gain of the FL server and the UAVs, $h_m =$ 5dBm. From Fig. \ref{communications}, we can observe that the communication time needed for the UAVs to transmit the local parameters to the FL server is much lower than that required by the IoV components, hence improving the communication efficiency in completing the FL task. Note that even without taking into account of the probability of link failure and the straggler's effect, the communication time required by the IoV components is much larger than that of the UAVs.

\begin{table}[t]
\caption{Simulation Parameters.} 
\label{simulated}
\centering
\begin{tabular}{p{5.5cm} | p{2.2cm} }
\hline \hline
\textbf{Parameter}& \textbf{Values}\\ [0.5ex]
 \hline 
Total Number of Cells, $I$ & 3\\
Total Number of UAVs, $M$ & 6\\
Minimum Threshold Level of Importance, $\zeta$ & 1\\
Sampling Rate of Worker, $s_{iw}$ & 250sps - 32000sps \\
Bandwidth of Worker, $b_{iw}$ \cite{chen2019joint} & 50kHz - 150kHz\\
Transmission Power of Worker, $p_{iw}$ \cite{chen2019joint} & 1mW - 10mW\\
Channel Gain of Worker, $h_{iw}$ & 2dBm - 8dBm\\
Data Size of Worker, $z^{w}$ & 100 MB\\
Data Size of Cell, $z^c$ & 500 MB\\
Flying Energy per Unit Time, $e^f$ &1000mW\\
Flying Velocity of UAV, $v$ & 10m/s\\
Computational Coefficient, $\kappa$ \cite{latency2018} & $10^{-26}$\\
Computational Capability of UAV, $f_{m}$ & $10^{8}$\\
Total Number of CPU Cycles, $a_{m}$ & 1GHz - 3GHz\\
Hovering Energy, $e^h_{m}$ & 5 J\\
Circuit Energy, $e^t_{m}$ & 5 J\\
Bandwidth of UAV, $b_m$ & 200kHz - 400kHz\\
Transmission Power of UAV, $p_{mt}$ \cite{uavpower} & 0.5W - 5W\\
Receiving Power of UAV, $p_{mr}$ & 0.5W - 1W\\
Channel Gain of UAV, $h_m$ & 5dBm - 25dBm\\
Cooperation Cost, $K_m$ & 2\\
Price coefficient of UAV per unit energy, $\delta_m$ & 0.03\\
Data Size of Model Owner, $z_{BS}$ & 1500 MB\\
Transmission Power of Model Owner, $p_{BS}$ \cite{uavpower} & 10W\\
Bandwidth of Model Owner, $b_{BS}$ & 5000kHz \\
Number of FL Iterations, $\mu$ & 20\\

\hline 
\end{tabular}
\end{table}

\begin{figure}[t]
\includegraphics[width=\linewidth]{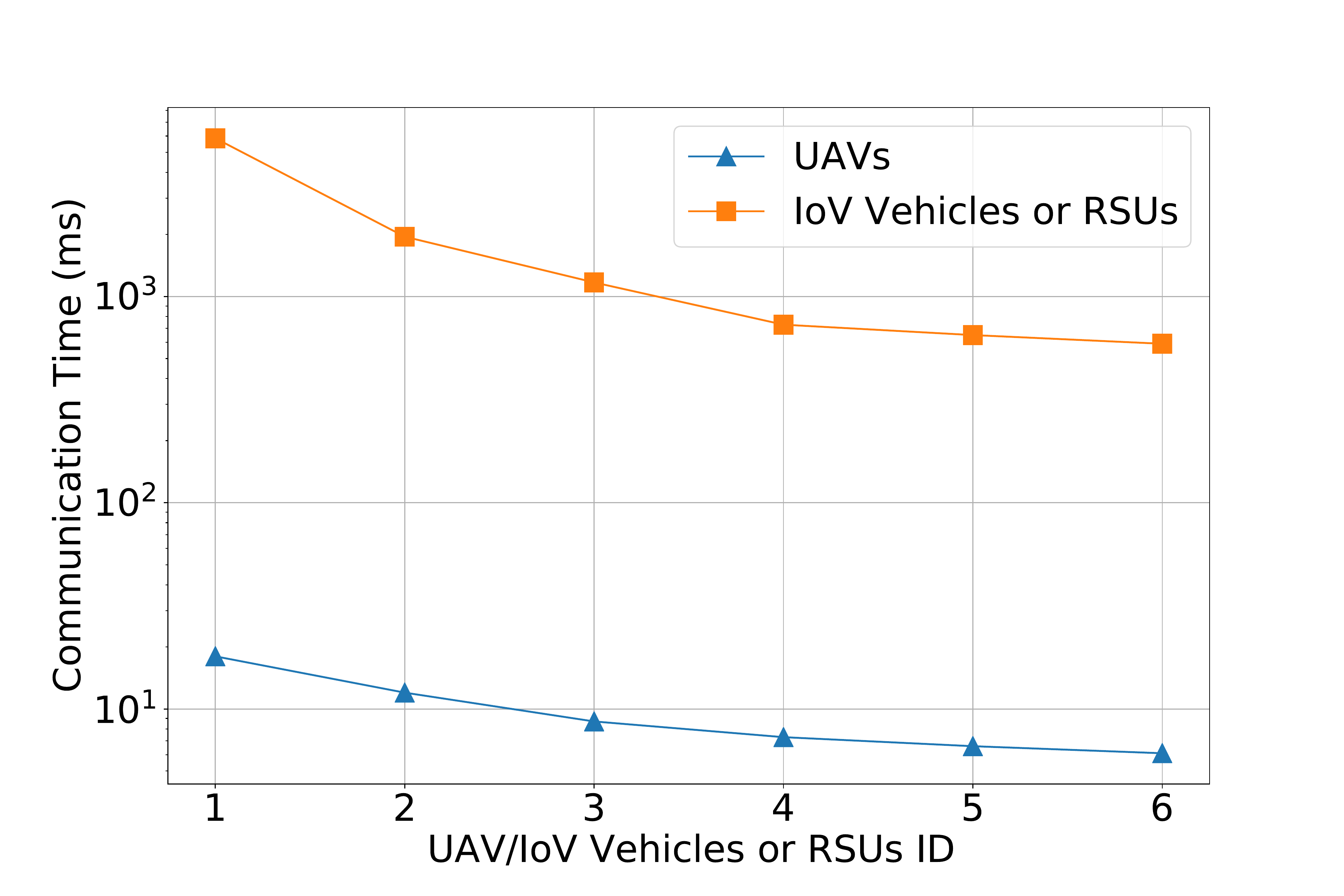}
\caption{Communication Time Needed by UAVs and IoV Components.}
\label{communications}
\end{figure}

\subsection{Preference of Cells of Workers}
\label{subsec:cellpreference}

\begin{table}[t]
\caption{Preference of Cells for Different Coalitions.} 
\centering
\renewcommand\arraystretch{1.5}
\begin{tabular}{c c c c}
\hline\hline
Cell ID& Coordinates & Importance & Top Preference of Cells\\ [0.5ex]
\hline 
Cell 1 & (200, 300) & 18.4& $\{3\},\{1,3\},\{2,3\}$ or $\{1,2,3\}$\\ 
Cell 2 & (500, 700) & 25.2 &$\{6\},\{2,6\}$\\
Cell 3 & (800, 600) & 34.6 &$\{6\}$\\ 
\hline
\end{tabular}
\label{cellpref}
\end{table}

In order to analyze the preference of cells for the UAVs and to determine their valuations in the bidding process, there are three factors that need to be taken into account, i.e., the importance of the cell, the time required by the UAVs to reach the cell, as well as the ability of the UAVs to complete the FL task. We consider 3 cells where Cell 1, Cell 2 and Cell 3 are located at coordinates $(200, 300)$, $(500, 700)$ and $(800, 600)$, respectively, as presented in second column of Table \ref{cellpref}. Column 4 of Table \ref{cellpref} shows the top preference of each cell.

Firstly, the cells submit their bids based on their own importance. The more important they are, the higher they earn from the payment by the model owner, and thus have higher incentive to employ UAVs to complete the FL task. Based on the sampling rate of the vehicles or RSUs, the importance of the vehicles or RSUs to the FL server is also calculated. Cell~3 is the most important cell with an importance value of 34.6 whereas Cell 1 is the least important cell with an importance value of 18.4. We assume that the price that the model owner pays for each unit of importance is the same for all 3 cells, i.e., $q_1=q_2=q_3=3$. Secondly, the valuations of each cell also depend on the time required by the UAVs to reach the cell. Cells of workers prefer UAVs that are nearer as they receive higher payment by the model owner for completing the task in a shorter period of time. We assume that the weight parameters for both cell importance and travelling time for the UAVs to the cells to be equal, i.e., $\theta_1=\theta_2=0.5$. The 6 UAVs are randomly distributed, among which their coordinates are presented in column 2 of Table \ref{uavpref}. We can calculate the Euclidean distance between cells and UAVs based on the coordinates of them. Since we assume that the UAVs travel at constant velocity $v$, the time required by UAV $m$ to reach cell $i$, denoted as $t^m_i$, can be computed by dividing distance over velocity, i.e., $t^m_i=\frac{d^m_i}{v}$. Thirdly, each cell will only submit bids to the UAVs that can complete the FL task, i.e., fulfil the requirement of completing a certain number of iterations as denoted by the model owner. 

\begin{table}[t]
\caption{Preference of UAVs for Different Cells.} 
\centering
\renewcommand\arraystretch{1.5}
\begin{tabular}{c c c c }
\hline\hline
UAV ID& Coordinates & Energy Capacity (Joules)  & Preference\\ 
\hline 
UAV 1 & $(100, 100)$ & 200 & -\\ 
UAV 2 & $(300, 500)$ & 500 & -\\
UAV 3 & $(100, 600)$ &1000 &Cell 3\\ 
UAV 4 & $(600, 200)$ & 1250 & Cell 3\\ 
UAV 5 & $(900, 200)$ & 3000 &Cell 3\\ 
UAV 6 & $(800, 800)$ & 3500 & Cell 3\\ 
\hline
\end{tabular}
\label{uavpref}
\end{table}

\begin{table}[t]
\caption{Preference of Cells for Individual UAVs.} 
\centering
\renewcommand\arraystretch{1.5}
\begin{tabular}{c c c c }
\hline\hline
UAV ID& Cell 1 & Cell 2  & Cell 3\\ [0.5ex]
\hline 
UAV 1 & 0& 0 & 0\\ 
UAV 2 & 0 & 0 & 0\\
UAV 3 & \textbf{43.4}&49.9 & 59.0\\ 
UAV 4 & 39.7 & 47.6 & 63.1\\ 
UAV 5 & 34.7 & 45.6 &64.2\\ 
UAV 6 & 34.0& \textbf{53.6} & \textbf{76.9}\\
\hline
\end{tabular}
\label{uavprefind}
\end{table}

To illustrate the preference analysis of the cells of workers, we first consider the scenario in which each UAV facilitates the FL training in the cells individually, i.e., no coalition is formed yet. Table \ref{uavprefind} illustrates the preference of each cell for each UAV. Cell 1 prefers UAV 3 the most. For Cell 2 and Cell 3, it is clearly seen that both cells prefer UAV 6. This is mainly because UAV 3 is nearest to Cell 1, whereas UAV 6 is nearest to Cell 2 and Cell 3. In addition to that, all cells do not submit any valuation for both UAV 1 and UAV 2 as they cannot individually fulfil the number of iterations required for the FL task. This is because UAV 1 and UAV 2 have low energy capacities (as shown in Table \ref{uavpref}), thereby they are not able to support any cell individually. 

However, the allocation of UAVs to the cells solely based on the valuations submitted by the cells may not benefit the UAVs. Taking Cell 2 as an example. If only a single UAV is considered in supporting a cell, Cell 2 will lose out to Cell~3 in the bidding process for UAV 6 since Cell 3 bids higher for UAV 6. As a result, Cell 2 is left with UAV 3, UAV 4 and UAV 5 to bid for. Among the remaining choices, Cell 2 prefers UAV 3 the most. Thus, Cell 1 will lose out to Cell~2 in bidding for UAV 3 and thus, with UAV 4 as the most preferred choice among the remaining UAVs. It may not be economically efficient for UAV 4 to be allocated to Cell 1 as it is far away from Cell 1. The energy needed for UAV 4 to travel to Cell 1 and back to its original position may offset the profit earned in supporting the FL training process in Cell~1. On the other hand, it may be more cost-effective for UAV 1 and UAV 3 to cooperate to support Cell 1 since the energy needed to reach Cell 1 is smaller.

We then consider coalition formation between two or more UAVs. Previously, when only individual UAVs are considered, UAV 1 and UAV 2 are not considered due to the limited energy capacities in completing the training process individually. However, when coalitions can be formed, UAV 1 and UAV~2 are now considered where either of them forms a coalition with UAV 3 or both of them form a coalition with UAV 3. From the perspective of Cell 1, it is indifferent among $\{3\}$, coalitions $\{1,3\}$, $\{2,3\}$ and $\{1,2,3\}$. The valuations of Cell 1 for these coalitions are the same as the time needed for all the UAVs in the respective coalitions to reach Cell 1 is equal and they are capable of completing the FL task. Note that the time needed for all UAVs in a coalition to reach a cell is determined by the UAV that is farthest from the cell. Similarly, Cell 2 is indifferent between $\{6\}$ and coalition $\{2,6\}$. Cell 3 prefers UAV 6 since it is capable of completing the task individually and it is nearest to Cell 3. 

However, there are two challenges to solve in determining the allocation of UAVs to different cells of workers. Firstly, two cells may prefer the same UAVs. For example, both Cell 2 and Cell 3 prefer UAV 6 but it is not possible for UAV 6 to support both cells at the same time. Secondly, it may not be profitable for some UAVs to support certain cells even if they are preferred by the cells. For example, coalition \{1,2,3\} is equally preferred by Cell 1 as coalitions $\{1,3\}$, $\{1,2\}$ and $\{3\}$. However, it may not be profitable for coalition $\{1,2,3\}$ to support Cell 1 if a smaller coalition of UAVs can complete the task at a lower cost. 

As such, we discuss the profit-maximizing behaviour of the UAVs as follows.

\subsection{Profit-Maximizing Behavior of UAVs}
\label{subsec:profit-max}

To analyze the profit-maximizing behavior of the UAVs, we consider 6 UAVs with different levels of energy capacity, i.e., low, medium and high (as previously presented in third column of Table \ref{uavpref}).

Fig. \ref{iteration} shows the maximum number of iterations given the different energy capacities. It is obvious that the larger the energy capacity of the UAV, the larger the number of iterations that the UAV can support. Individually, UAV 1 and UAV 2 do not choose to support any cell as the possible profit earned in supporting any cell is negative. For UAV 1 and UAV 2, since the valuation bids of all cells are zero, it is natural not to support any cell as they are not earning any revenue. Meanwhile, the profit-maximizing property of the UAVs has resulted in the competition of UAV 3, UAV 4, UAV 5 and UAV 6 to facilitate FL training in the same cell, i.e., Cell 3. However, by allocating UAV 3, UAV 4, UAV 5 and UAV 6 to Cell 3, the revenue earned does not increase but the cost increases which contribute to the overall smaller profit earned. Thus, it is possible to allocate one of the UAVs to another cell such that the overall profit of the UAVs is higher, which is discussed in the next subsection.

\begin{table}[h]
\caption{Revenue, Cost of Profit for Different Coalitions in Each Cell.} 
\centering
\renewcommand\arraystretch{1.5}
\begin{tabular}{|p{1cm}|p{0.35cm}p{0.35cm}p{0.6cm}|p{0.35cm}p{0.35cm}p{0.35cm}|p{0.35cm}p{0.35cm}p{0.35cm}|}
\hline

\multirow{2}{*}{Coalition} & \multicolumn{3}{c|}{Cell 1} & \multicolumn{3}{c|}{Cell 2} & \multicolumn{3}{c|}{Cell 3} \\
       \cline{2-10} 
    & Rev.   & Cost  & Prof. & Rev. & Cost & Prof.  & Rev. & Cost & Prof.     \\

    \hline\hline
$\{3\}$       & 43.4 & 23.5 & 19.9  & 49.9  & 16.2  & 33.7  & 59.0  & 21.5  & 37.5  \\
$\{1,3\}$ & 43.4 & 17.9 & \textbf{25.5}  & 44.7  & 22.4  & 22.4  & 57.7  & 23.1  & 34.6  \\
$\{2,3\}$       & 43.4 & 32.9 & 10.5 & 49.9  & 32.7  & 17.2   & 59.0  & 31.5  & 27.5  \\
$\{1,2,3\}$   & 43.4 & 26.8 & 16.6  & 44.7  & 40.0& 4.70  & 57.1  & 38.4  & 19.3\\

$\{6\}$ & 34.0 & 53.8 & -19.8 & 53.6  & 33.8&  19.8 & 76.9  & 34.8& \textbf{42.1} \\
$\{2,6\}$   & 34.0 &  34.3   & -0.33  & 53.6 &  32.4 & 21.2  & 61.7  & 40.3& 21.4  \\
\hline

\end{tabular}
\label{profit}
\end{table}

Given the varied energy cost of the UAVs and different revenue from different cells, each UAV coalition earns different profits for supporting different cells of workers. From Table~\ref{profit}, we observe that although the valuations for coalitions~$\{3\}$, $\{1,3\}$, $\{2,3\}$ and $\{1,2,3\}$ by Cell 1 are the same, the costs are different. As a result, the profit of coalition~$\{1,3\}$ is the highest if it is allocated to Cell 1. The addition of UAV 2 to the coalition only incurs extra cost, which leads to a lower profit. Hence, it is not economically viable for coalition $\{1,2,3\}$ to be allocated to Cell 1. Similarly, for Cell~2, coalitions $\{2,6\}$ and $\{6\}$ have same valuation. It costs lower for Cell 2 to be supported by a coalition of UAV 2 and UAV 6 instead of UAV 6 individually. 

With the different revenue and cost structures of the UAVs, we next discuss the auction-based UAV-cell allocation.

\begin{figure}
\includegraphics[width=\linewidth]{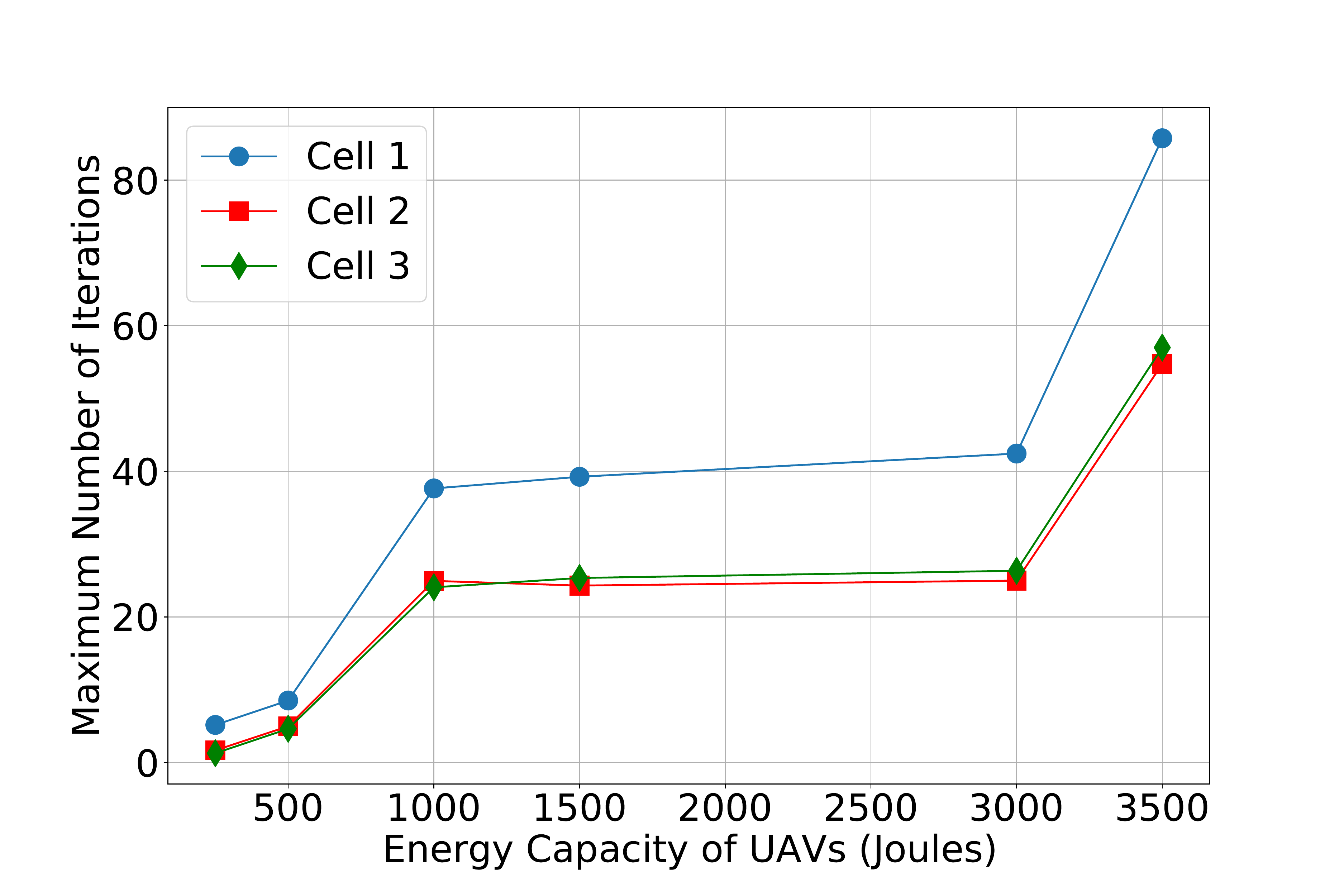}
\caption{Maximum Number of Iterations under Different Energy Capacities.}
\label{iteration}
\end{figure}

\subsection{Allocation of UAVs to Cells of Workers}
Given the valuations of cells for different coalitions of UAVs through the auction, the UAVs as the auctioneers need to decide on the allocation of UAVs to the different cells of workers such that the total profit earned by all UAVs is maximized and the preferences of the cells of workers are taken into account.

\begin{figure}
\centering
\includegraphics[width=0.8\linewidth]{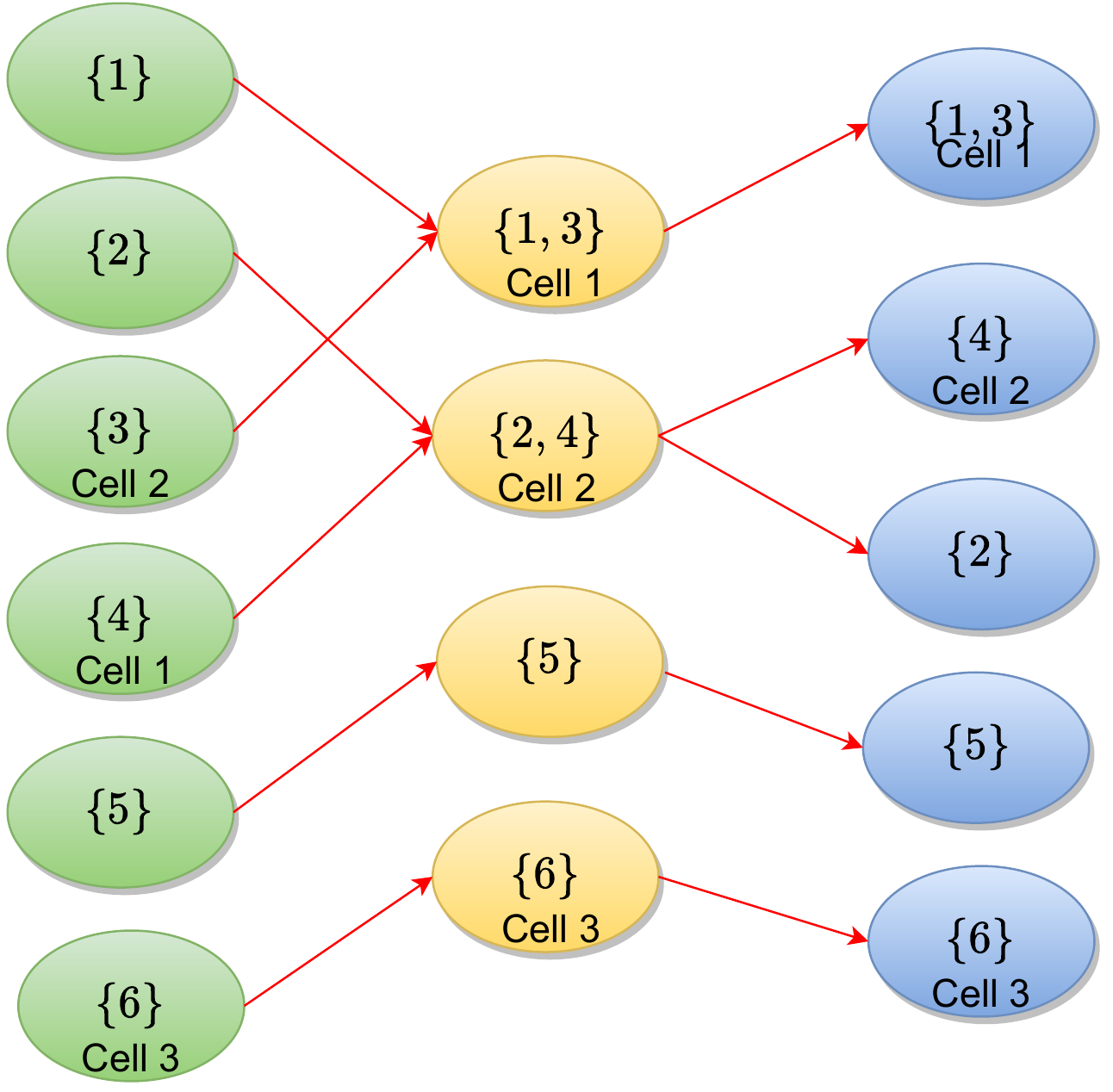}
\caption{Illustration of Merge-and-Split Mechanism and Allocation of UAVs to Cells of Workers.}
\label{merge}
\end{figure}


The merge-and-split algorithm is used to determine the partition of the UAVs that maximizes the total profit. Fig.~\ref{merge} demonstrates the merge-and-split mechanism that determines the partition that increases the total profit earned. When each UAV is considered individually where no coalition is formed, the maximum total attainable profit is achieved when UAV 4 is allocated to Cell 1, UAV~3 is allocated to Cell~2 while UAV~6 is allocated to Cell~3. UAV~1, UAV~2 and UAV~5 are not allocated to any cell. The maximum total attainable profit is increased when the partition changes to $\{\{1,3\},\{2,4\},\{5\},\{6\}\}$ through a merge mechanism, where UAV 1 and UAV 3 are allocated to Cell 1, UAV 2 and UAV 4 are allocated to Cell 2 and UAV 6 is allocated to Cell 3. Then, the split mechanism results in a change in partition again to  $\{\{1,3\},\{2\},\{4\},\{5\},\{6\}\}$, where coalition $\{1,3\}$ is allocated to Cell 1, UAV 4 is allocated to Cell 2 and UAV~6 is allocated to Cell 3. A change in partition will only occur when it results in an increase in total profit earned given the optimal allocation of UAVs to cells.

Although UAV 5 is positively valued by all 3 cells of workers, it is not allocated to any cell. From the perspective of Cell~1, coalition $\{1,3\}$ can support the FL task at a lower cost than that of UAV~5. Similarly, the costs of UAV~4 and UAV 6, if they are allocated to Cell 2 and Cell 3 respectively, are lower than that if any of the cell is supported by UAV~5. Moreover, UAV 2 is also not allocated to any cell. As mentioned in Section \ref{subsec:profit-max}, the profit decreases when UAV 2 joins coalition~$\{1,3\}$ in supporting Cell 1, hence it is more profitable for coalition $\{1,3\}$ to be allocated to Cell 1, instead of coalition $\{1,2,3\}$. UAV 2 also does not join UAV~4 in supporting Cell~2. Although UAV~2 is nearer to Cell 2, it is not capable of completing the FL task individually. The valuations of Cell 2 for UAV 4 and coalition $\{2,4\}$ are the same. However, it is more cost-effective for Cell 2 to be supported only by UAV 4 instead of both UAV 2 and UAV 4. As such, we see that the grand coalition, where all UAVs join a single coalition, is not stable, thus validating our proof in Section \ref{subsec:coalitionalgorithm}. In fact, by forming the grand coalition, it can only be allocated to Cell~3 with a profit of 2.12 as the UAVs will not earn positive profit if they are allocated to either Cell~1 or Cell~2.

\begin{figure}
\centering
\includegraphics[width=\linewidth]{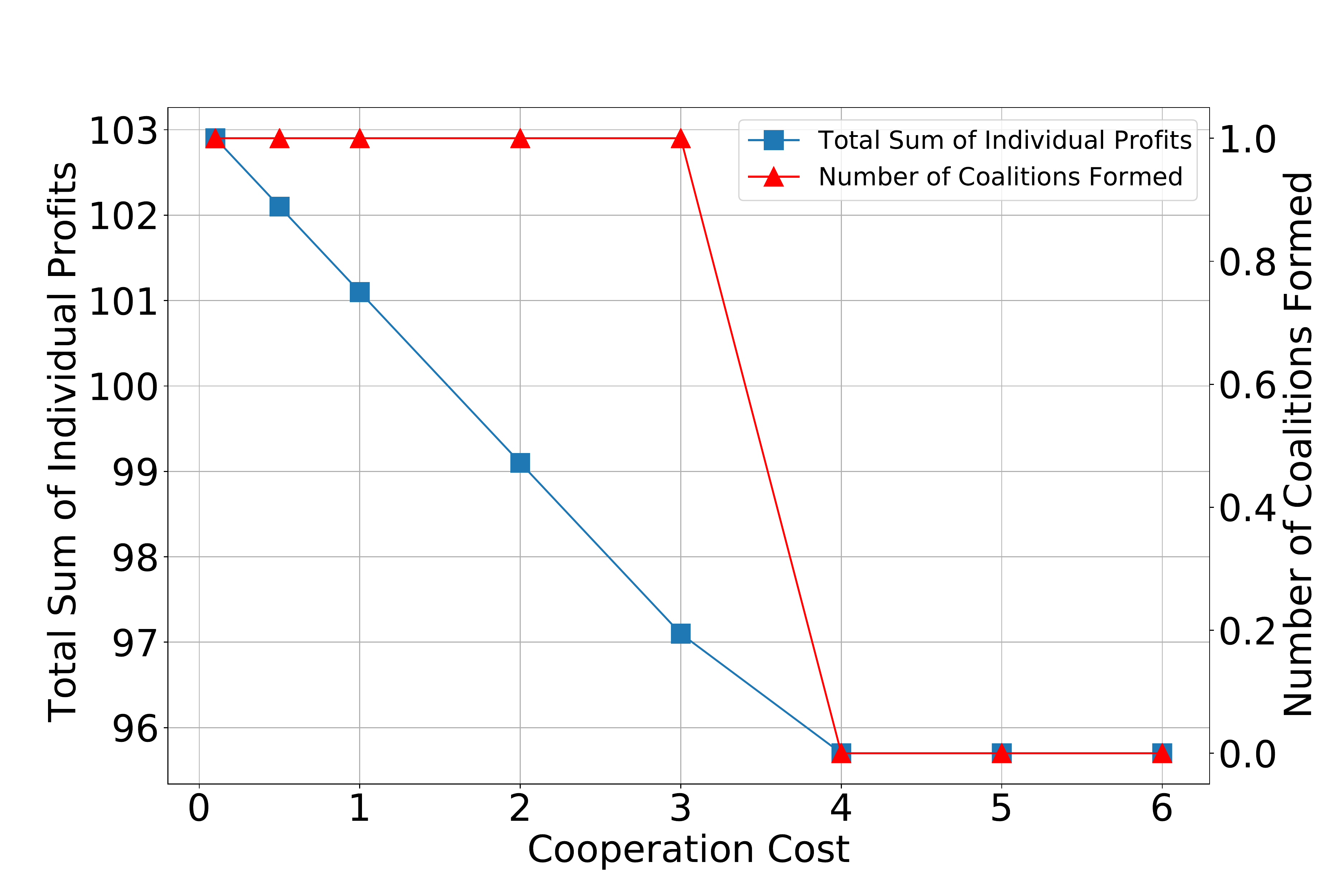}
\caption{Total Profit and Number of Coalitions vs Cooperation Cost.}
\label{cooperation}
\end{figure}

Fig. \ref{cooperation} shows that the total profit of the UAVs decreases as the cooperation cost among the UAVs increases. When the cooperation cost is more than or equal to 4, the UAVs prefer to support the FL training process in the cells individually and not to form any coalition where UAV~3, UAV~4 and UAV~6 are allocated to Cell~2, Cell~1 and Cell~6 respectively. In other words, it is not cost-effective for the UAVs to form coalition anymore. By not forming any coalition, the UAVs are able to earn a profit of 95.7, which is higher than that if the UAVs decide to form coalitions when the cooperation cost is more than or equal to 4.
\begin{figure}
\centering
\includegraphics[width=\linewidth]{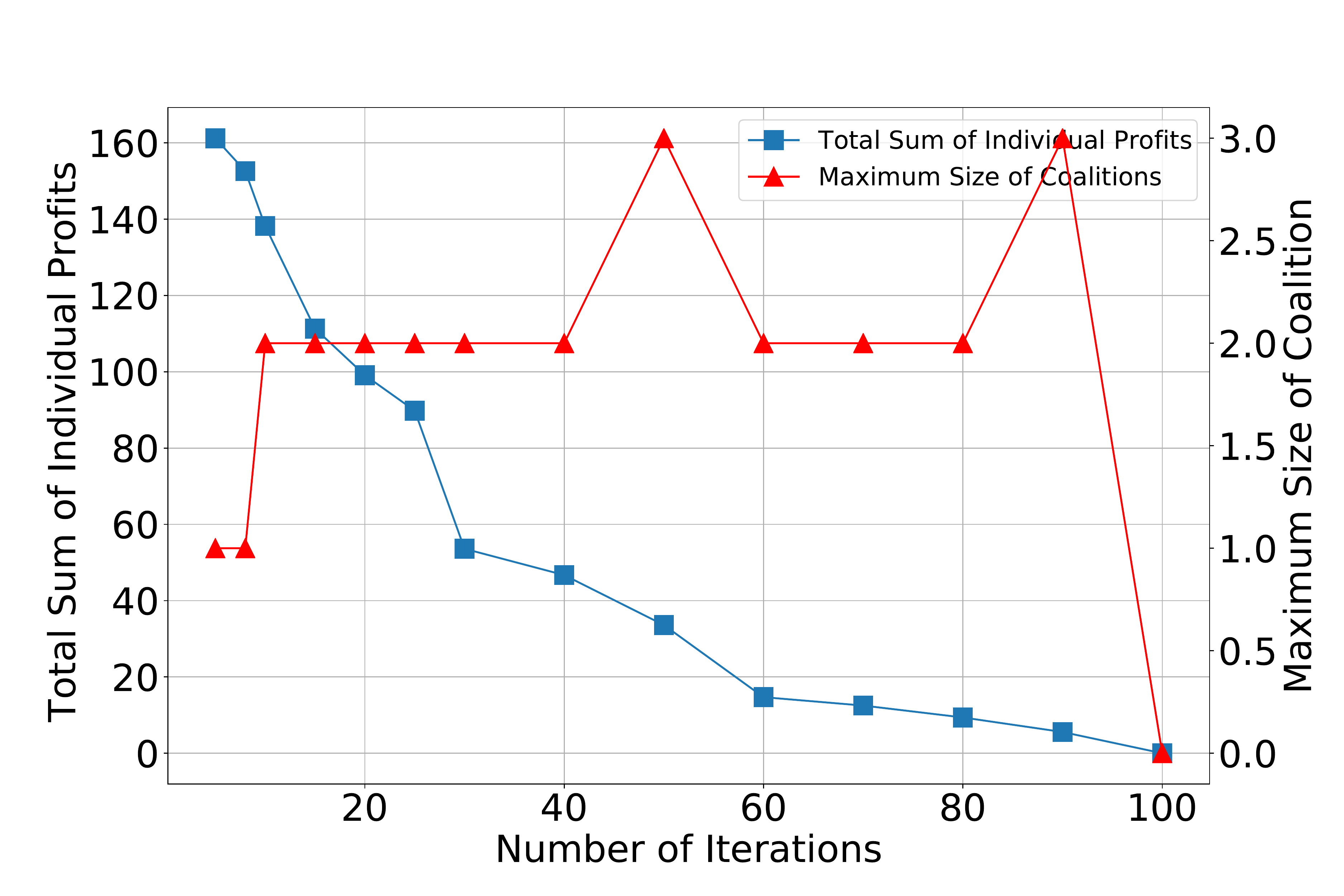}
\caption{Total Profit and Size of Coalitions vs Number of Iterations.}
\label{numiterate}
\end{figure}

Furthermore, the total profit earned by the UAVs is also affected by the number of iterations announced by the model owner. As the number of iterations required to be completed by the UAVs increases, the total profit earned by the UAVs decreases as seen in Fig.~\ref{numiterate}. Clearly, since the payment price of the cells of workers does not change and the UAVs incur more energy to facilitate the FL task, the total profit earned decreases. From Fig.~\ref{numiterate}, we observe that the maximum size of coalitions formed is 3 when the number of iterations required is 50. This implies that by forming a coalition of size 3 and be allocated to a cell of workers, the UAV coalition is still able to earn a positive profit. When the number of iterations is 60 to 80, only coalition $\{1,6\}$ is allocated to Cell~1 where the profit decreases as the number of iterations increases. The rest of the UAVs are allocated to neither Cell~2 nor Cell~3 as the profits earned by any possible coalitions from supporting any of these cells are negative. When the number of iterations required by the FL task is 90, the coalition $\{1,3,6\}$ is allocated to Cell~1. The UAVs prefer not to support any cell of workers, and thus there is no need to form a coalition when the number of iterations is 100. Therefore, a coalition with size of larger than 3 does not form due to two reasons. Firstly, it is possible to form a smaller coalition to support a cell of workers such that there is no need for a larger UAV coalition. Secondly, the cost incurred by a larger UAV coalition is larger than the revenue gained such that the UAVs are better off not supporting any cell of workers. 

Next, we compare the performance of the proposed joint auction-coalition formation framework against the existing schemes.

\begin{figure}
\centering
\includegraphics[width=\linewidth]{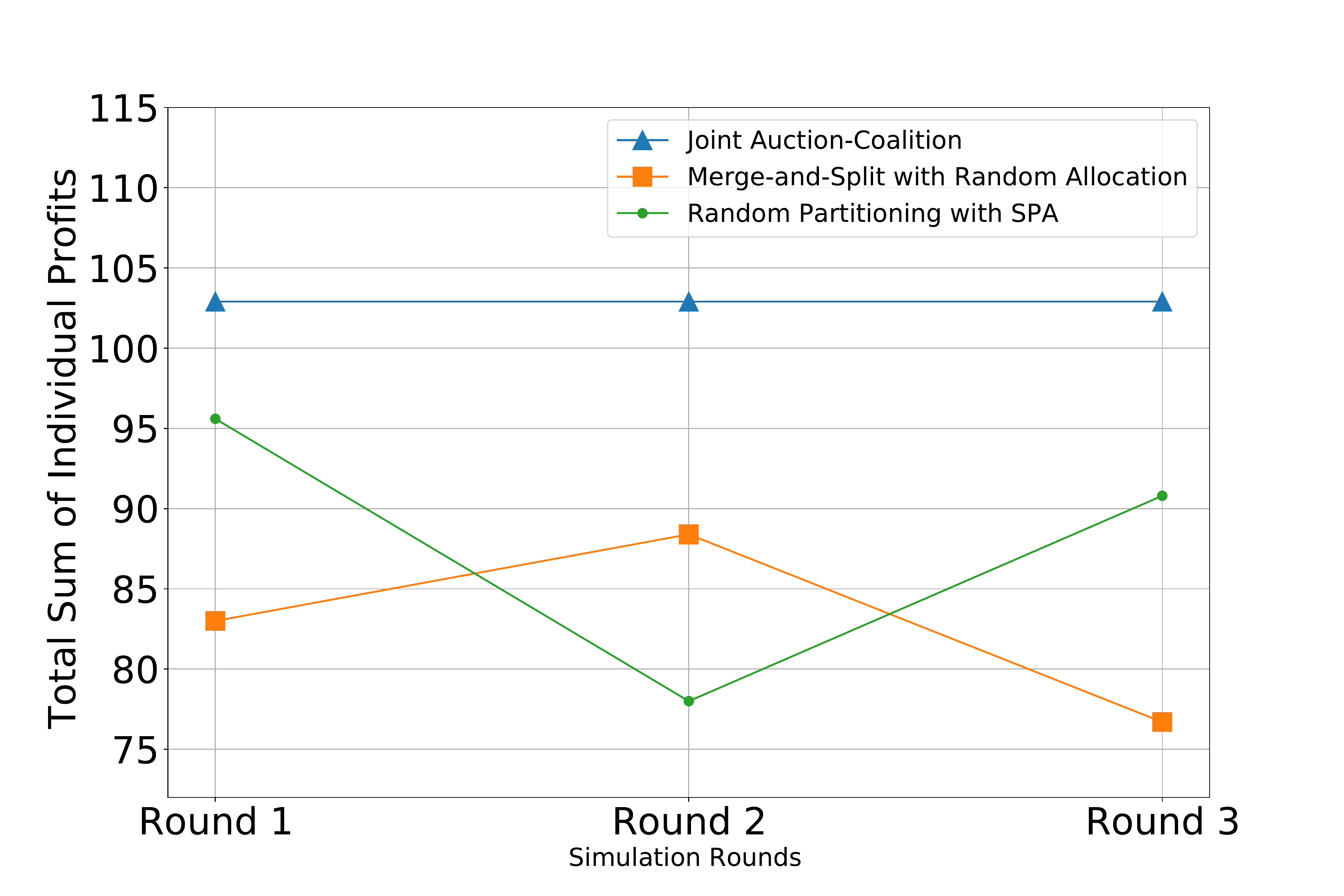}
\caption{Comparison with Existing Schemes.}
\label{comparison}
\end{figure}

\subsection{Comparison with Existing Schemes}
For comparison, we have chosen two existing schemes: (i) the merge-and-split algorithm to determine the formations of the UAV coalitions with random allocation to the IoV groups, and (ii) random partitioning of the UAV coalitions with second-price auction. We have run three rounds of simulations and the performances of the three schemes are shown in Fig.~\ref{comparison}.

The joint auction-coalition formation framework achieves the highest profit in all three rounds of simulations. Due to the randomness in either the partitioning or the allocation of the UAV coalitions, the UAVs are not able to maximize their profits. On one hand, by adopting the merge-and-split framework to determine the formations of the UAV coalitions and randomly allocating the UAV coalitions to the IoV groups, the profits of the UAV coalitions are not maximized as the UAV coalitions are not allocated to the IoV groups that value them most. On the other hand, by randomly deciding on the partitioning of the UAV coalitions and allocating them based on the second-price auction, the potential of the UAV coalitions to earn higher profits is not exploited, which is achievable by considering other forms of partitioning. 

\section{Conclusion}
\label{sec:conc}

In this paper, we have proposed a joint auction-coalition formation framework of UAVs to facilitate resource-constrained IoV components in completing the FL tasks. Firstly, we design an auction scheme where each cell of workers submit its bids to all possible coalitions of UAVs based on their importance of the cells and the distance between the cells and the UAVs. Then, we use the merge-and-split algorithm to decide on the optimal coalitional structure that maximizes the total profit of the UAVs. 

For our future work, we can consider more factors in determining the coalition formation of UAVs. For example, we can take into account the social properties of the UAVs when they form coalitions such that malicious UAVs can be identified and eliminated so that the FL performance is not adversely affected. Furthermore, we can consider a joint trajectory optimization and energy efficient coalition formation game in order to complete the FL tasks more efficiently. 

\bibliographystyle{ieeetr}
\bibliography{UAV-revision}

\begin{IEEEbiography}[{\includegraphics[width=1in,height=1.25in,clip,keepaspectratio]{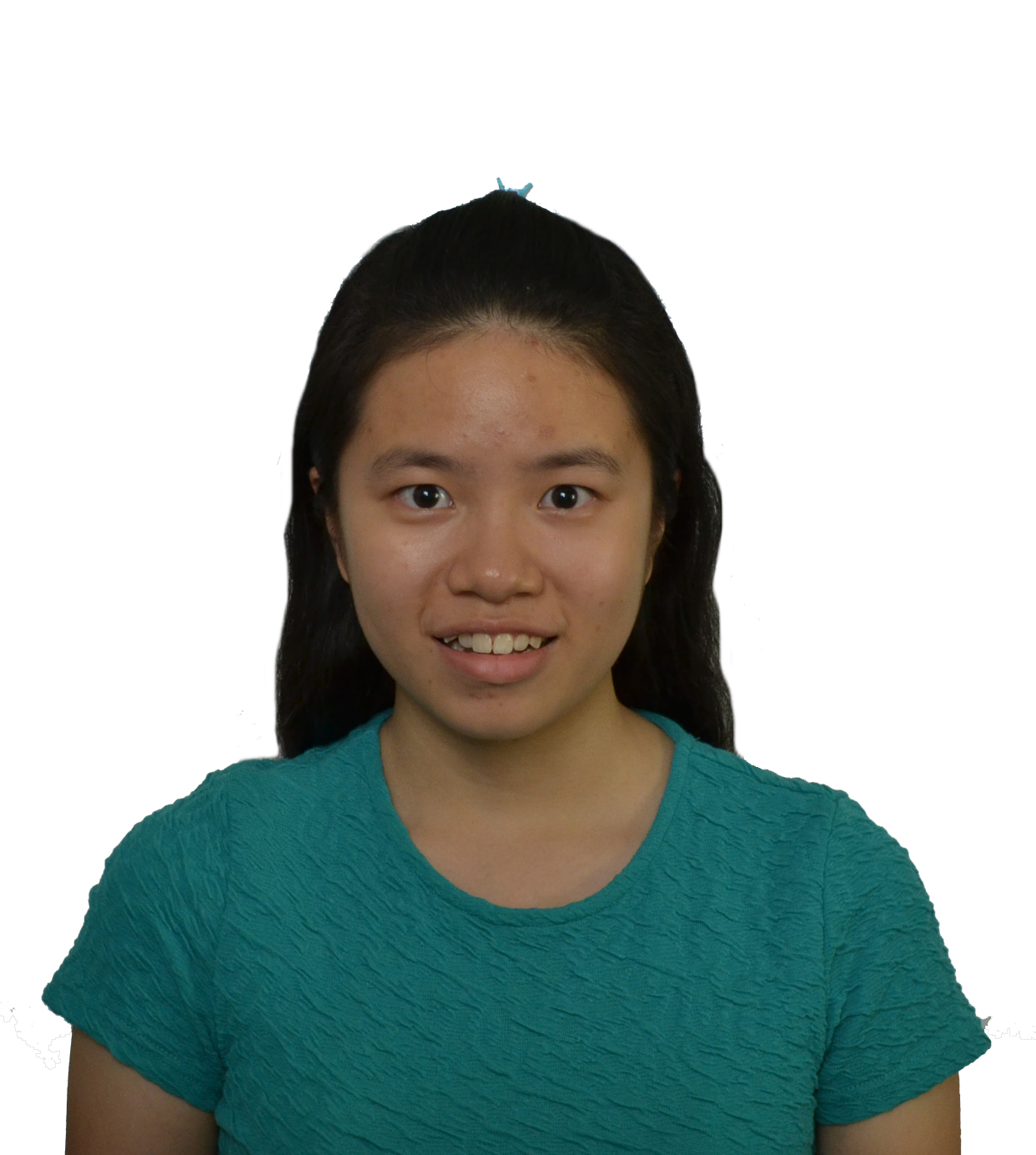}}]{Jer Shyuang Ng} graduated with Double (Honours) Degree in Electrical Engineering (Highest Distinction) and Economics from National University of Singapore (NUS) in 2019. She is currently an Alibaba PhD candidate with the Alibaba Group and Alibaba-NTU Joint Research Institute, Nanyang Technological University (NTU), Singapore. Her research interests include incentive mechanisms and edge computing.
\end{IEEEbiography}

\vspace{-1.5cm}

\begin{IEEEbiography}[{\includegraphics[width=1in,height=1.25in,clip,keepaspectratio]{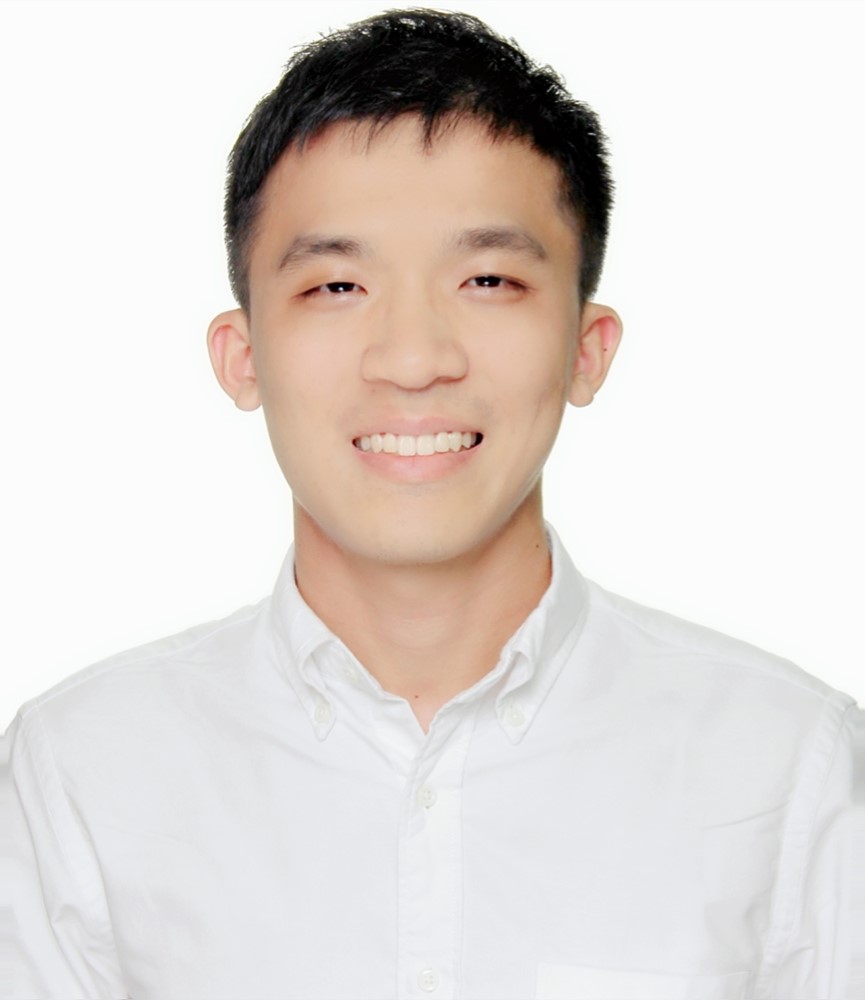}}]{Wei Yang Bryan Lim} graduated with double First Class Honours in Economics and Business Administration (Finance) from the National University of Singapore (NUS) in 2018. He is currently an Alibaba PhD candidate with the Alibaba Group and Alibaba-NTU Joint Research Institute, Nanyang Technological University, Singapore. His research interests include Federated Learning and Edge Intelligence.
\end{IEEEbiography}

\vspace{-1.5cm}

\begin{IEEEbiography}[{\includegraphics[width=1in,height=1.25in,clip,keepaspectratio]{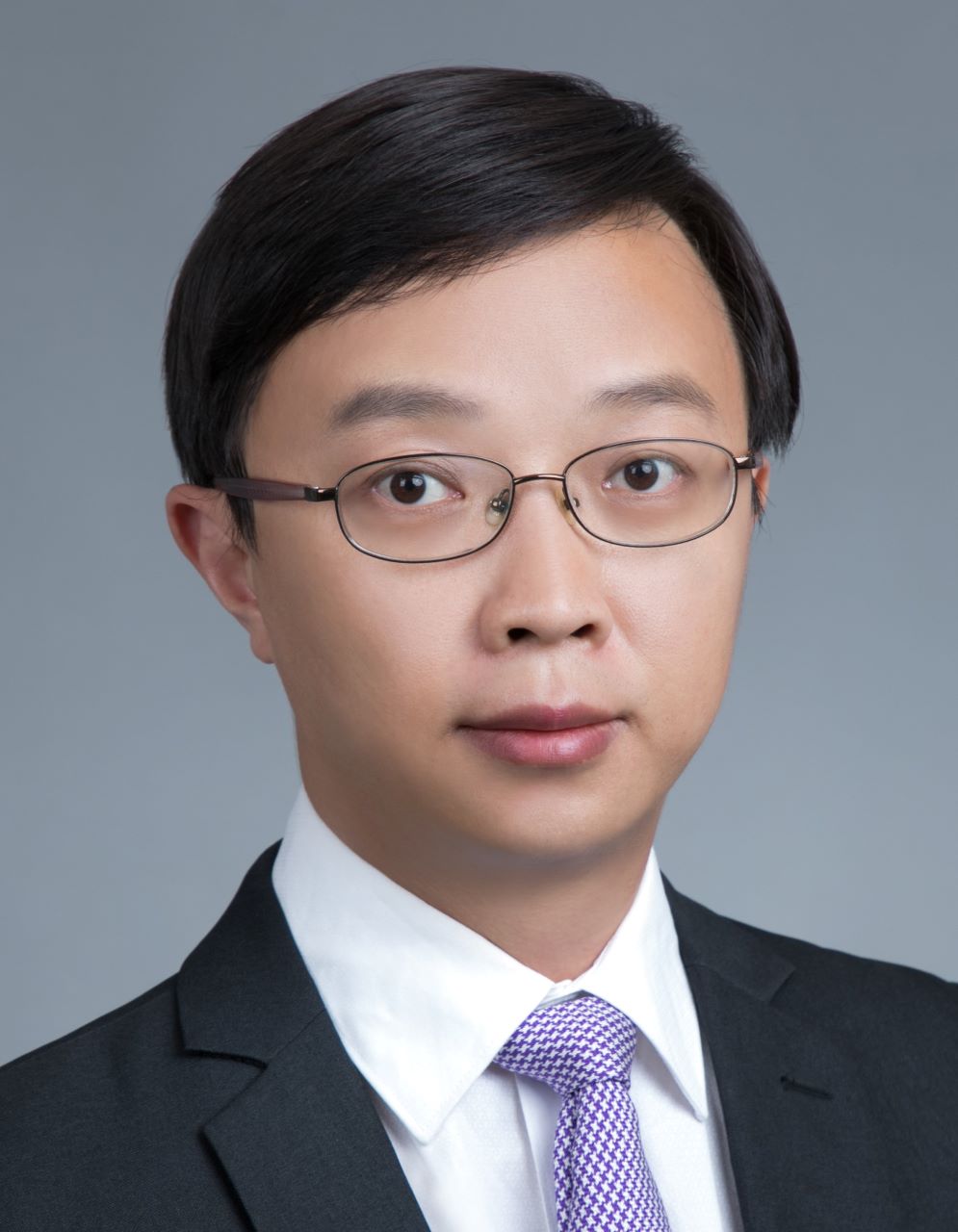}}]{Hong-Ning Dai} (Senior Member, IEEE) is currently with Faculty of Information Technology at Macau University of Science and Technology as an associate professor. He obtained the Ph.D. degree in Computer Science and Engineering from Department of Computer Science and Engineering at the Chinese University of Hong Kong. His current research interests include Internet of Things and blockchain technology. He has served as editors for Computer Communications (Elsevier), Connection Science (Taylor \& Francis),  IEEE Access, guest editors for IEEE Transactions on Industrial Informatics, IEEE Transactions on Emerging Topics in Computing. He is a senior member of the Institute of Electrical and Electronics Engineers (IEEE). 
\end{IEEEbiography}

\vspace{-1.5cm}

\begin{IEEEbiography}[{\includegraphics[width=1in,height=1.25in,clip,keepaspectratio]{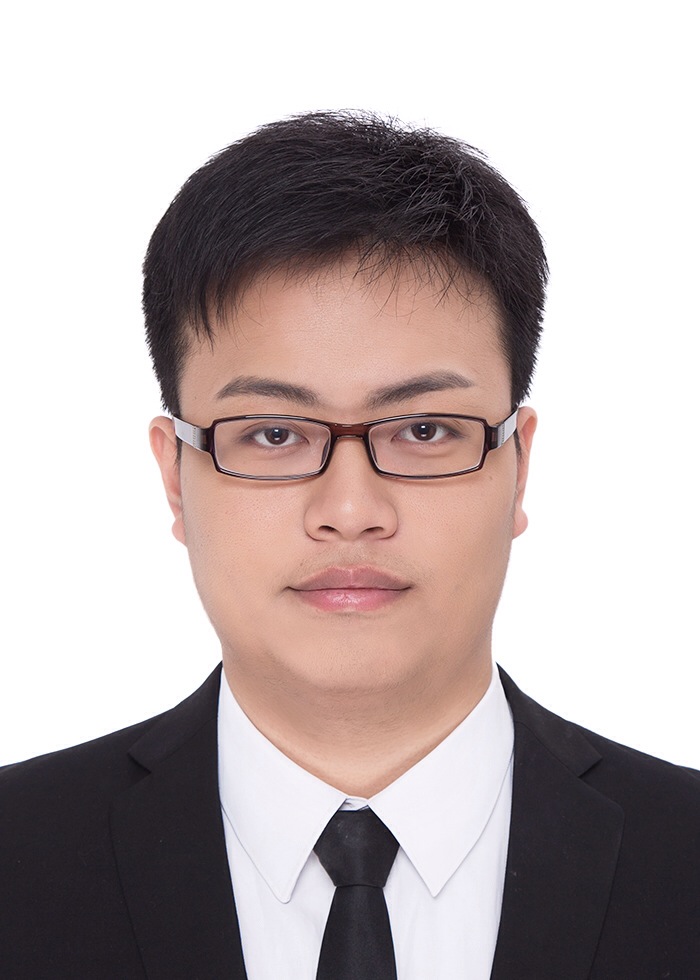}}]{Zehui Xiong}(S'17) received his B.Eng degree with the highest honors in Telecommunication Engineering from Huazhong University of Science and Technology, Wuhan, China, in Jul 2016. From Aug 2016 to Oct 2019, he pursued the Ph.D. degree in the School of Computer Science and Engineering, Nanyang Technological University, Singapore. Since Nov 2019, he has been with Alibaba-NTU Singapore Joint Research Institute. He was a visiting scholar with Department of Electrical Engineering at Princeton University from Jul to Aug 2019. He was also a visiting scholar with BBCR lab in Department of Electrical and Computer Engineering at University of Waterloo from Dec 2019 to Jan 2020. His research interests include resource allocation in wireless communications, network games and economics, blockchain, and edge intelligence.
\end{IEEEbiography}

\vspace{-1.5cm}

\begin{IEEEbiography}[{\includegraphics[width=1in,height=1.25in,clip,keepaspectratio]{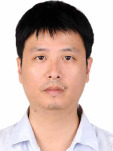}}]{Jianqiang Huang} is currently the Director of the Alibaba DAMO Academy. His research interests focus on visual intelligence in the city brain project of Alibaba. He received the second prize of National Science and Technology Progress Award in 2010.
\end{IEEEbiography}

\vspace{-1.5cm}

\begin{IEEEbiography}[{\includegraphics[width=1in,height=1.25in,clip,keepaspectratio]{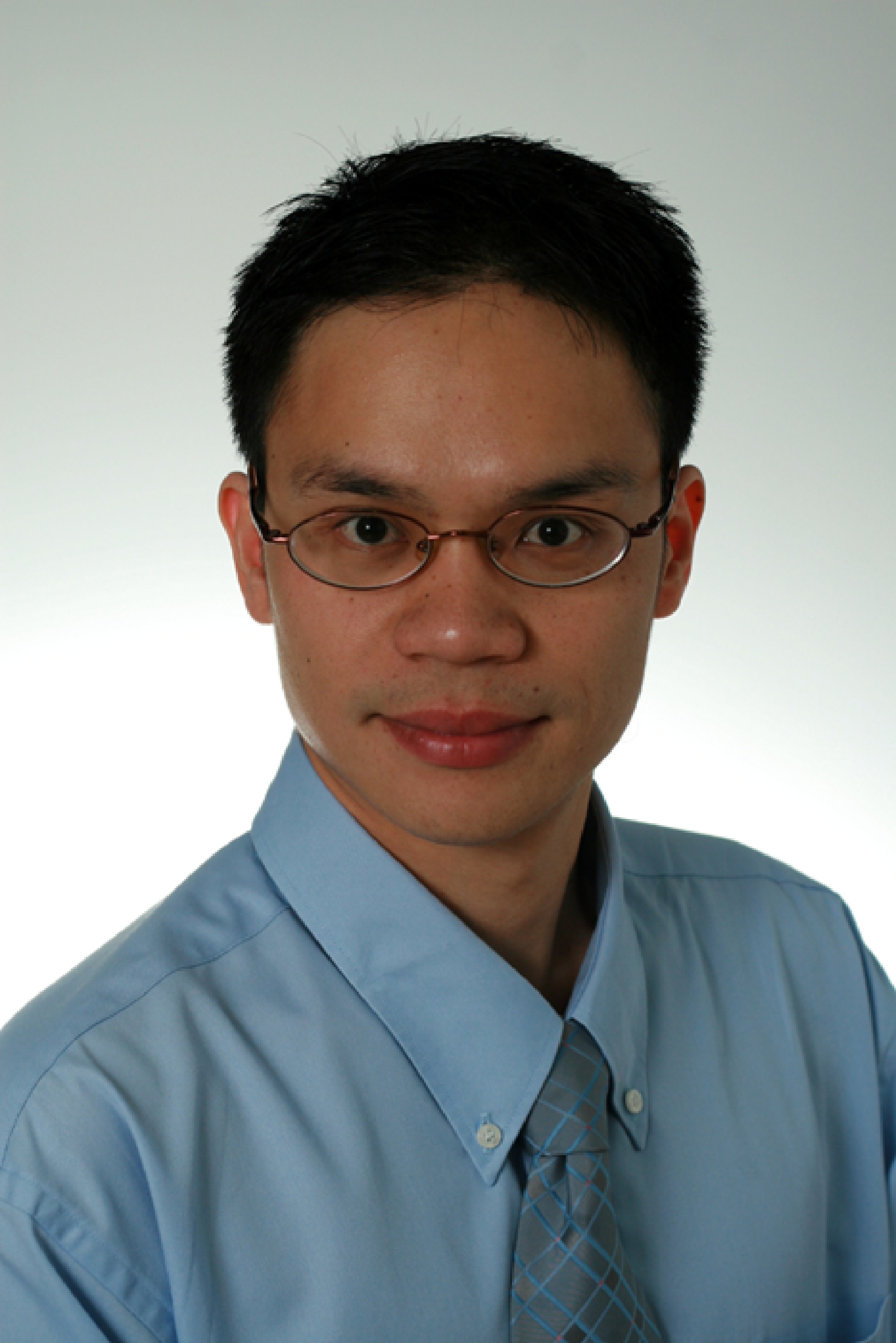}}]{Dusit Niyato} (M'09-SM'15-F'17) is currently a professor in the School of Computer Science and Engineering and, by courtesy, School of Physical \& Mathematical Sciences, at the Nanyang Technological University, Singapore. He received B.E. from King Mongkuk’s Institute of Technology Ladkrabang (KMITL), Thailand in 1999 and Ph.D. in Electrical and Computer Engineering from the University of Manitoba, Canada in 2008. He has published more than 380 technical papers in the area of wireless and mobile networking, and is an inventor of four US and German patents. He has authored four books including “Game Theory in Wireless and Communication Networks: Theory, Models, and Applications” with Cambridge University Press. He won the Best Young Researcher Award of IEEE Communications Society (ComSoc) Asia Pacific (AP) and The 2011 IEEE Communications Society Fred W. Ellersick Prize Paper Award. Currently, he is serving as a senior editor of IEEE Wireless Communications Letter, an area editor of IEEE Transactions on Wireless Communications (Radio Management and Multiple Access), an area editor of IEEE Communications Surveys and Tutorials (Network and Service Management and Green Communication), an editor of IEEE Transactions on Communications, an associate editor of IEEE Transactions on Mobile Computing, IEEE Transactions on Vehicular Technology, and IEEE Transactions on Cognitive Communications and Networking. He was a guest editor of IEEE Journal on Selected Areas on Communications. He was a Distinguished Lecturer of the IEEE Communications Society for 2016-2017. He was named the 2017, 2018, 2019 highly cited researcher in computer science. He is a Fellow of IEEE.
\end{IEEEbiography}

\vspace{-1.5cm}

\begin{IEEEbiography}[{\includegraphics[width=1in,height=1.25in,clip,keepaspectratio]{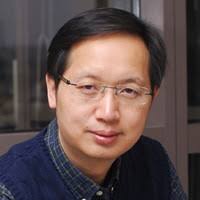}}]{Xiansheng Hua} (Fellow, IEEE) received the B.S. and Ph.D. degrees in applied mathematics from Peking University, Beijing, China, in 1996 and 2001, respectively. He joined Microsoft Research Asia, Beijing, in 2001, as a Researcher. He was a Principal Research and a Development Lead in multimedia search with Microsoft Search Engine, Bing, Red- mond, WA, USA, from 2011 to 2013. He was a Senior Researcher with Microsoft Research Red- mond, Redmond, from 2013 to 2015. He became a Researcher and the Senior Director of Alibaba
Group, Hangzhou, China, in 2015, where he is also leading the Search Division, Visual Computing Team, Alibaba Cloud, and DAMO Academy. He is currently a Distinguished Engineer/Vice President of Alibaba Group, where he is also leading a team working on large-scale visual intelligence on cloud. He has authored or coauthored more than 200 research articles and has filed more than 90 patents. His research interests include big multimedia data search, advertising, understanding and mining, pattern recognition, and machine learning. He is also an ACM Distinguished Scientist. He was a recipient of the 2008 MIT Technology Review TR35 Young Innovator Award for his outstanding contributions on video search. He was also a recipient of the Best Paper Award from ACM Multimedia 2007 and the Best Paper Award of the IEEE TRANSACTIONS ON CIRCUITS AND SYSTEMS FOR VIDEO TECHNOLOGY in 2014. He has served as the Program Co-Chair for the IEEE ICME 2012, ACM Multimedia 2012, and the IEEE ICME 2013. He will be serving as the General Co-Chair for ACM Multimedia in 2020.
\end{IEEEbiography}
\vspace{-1.5cm}

\begin{IEEEbiography}[{\includegraphics[width=1in,height=1.25in,clip,keepaspectratio]{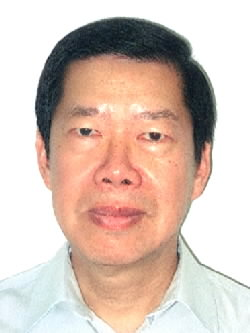}}]{Cyril Leung} received the B.Sc. (First Class Hons.) degree from Imperial College, University of London, U.K., and the M.S. and Ph.D. degrees in electrical engineering from Stanford University. He has been an Assistant Professor with the Department of Electrical Engineering and Computer Science, Massachusetts Institute of Technology, and the Department of Systems Engineering and Computing Science, Carleton University. Since 1980, he has been with the Department of Electrical and Computer Engineering, University of British Columbia (UBC), Vancouver, Canada, where he is a Professor and currently holds the PMC-Sierra Professorship in Networking and Communications. He served as an Associate Dean of Research and Graduate Studies with the Faculty of Applied Science, UBC, from 2008 to 2011. His research interests include wireless communication systems, data security and technologies to support ageless aging for the elderly. He is a member of the Association of Professional Engineers and Geoscientists of British Columbia, Canada.
\end{IEEEbiography}

\vspace{-1.5cm}

\begin{IEEEbiography}[{\includegraphics[width=1in,height=1.25in,clip,keepaspectratio]{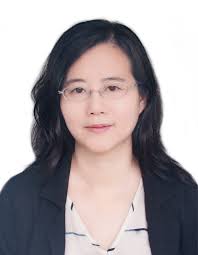}}]{Chunyan Miao} received the BS degree from Shandong University, Jinan, China, in 1988, and the MS and PhD degrees from Nanyang Technological University, Singapore, in 1998 and 2003, respectively. She is currently a professor in the School of Computer Science and Engineering, Nanyang Technological University (NTU), and the director of the Joint NTU-UBC Research Centre of Excellence in Active Living for the Elderly (LILY). Her research focus on infusing intelligent agents into interactive new media (virtual, mixed, mobile, and pervasive media) to create novel experiences and dimensions in game design, interactive narrative, and other real world agent systems.
\end{IEEEbiography}

\end{document}